\documentclass[11pt]{amsart}
\usepackage{amsfonts,amssymb,graphicx,mystyle,color,bbm}
\usepackage[colorlinks,linkcolor=blue]{hyperref} 
\usepackage{tikz}
\usepackage{cite}
\usepackage{hyperref}
\usepackage[multiple]{footmisc}
 \let\backslash=\setminus \let\ge=\geqslant \let\le=\leqslant 
\def\a{\alpha} \def\b{\beta} \def\d{\delta} \def\e{\varepsilon}  \def\g{\gamma} \def\l{\lambda} \def\s{\sigma} \def\t{\tau} \def\z{\zeta}
\def\D{\Delta}   \def\S{\Sigma}
     \def\N{\mathcal N} 

\def\P{\mathcal P} \def\T{\mathcal T}  \def\X{\mathcal X}
\def\Re{\mathbb{R}}

\def\rho{\varrho}

\def\BR{\text{BR}}
\def\Graph{\text{Graph}}

\usetikzlibrary{calc}
\usepackage{amssymb}
\usepackage{tikz-cd}
\usepackage{enumitem}
\usepackage{setspace}
\usepackage{float}
\usepackage{commath}
\usepackage{caption}
\usepackage{subcaption}
\usepackage{multirow,array}

\begin{document}\openup 1\jot

\begin{abstract} 
	We present the following analog of O'Neill's Theorem (Theorem 5.2 in \cite{O1953}) for finite games. 
	Let $C_1, \ldots, C_k$ be the components of Nash equilibria of a finite normal-form game $G$.  For each $i$, let $c_i$ be the index of $C_i$.  For each $\e>0$,  there exist pairwise disjoint neighborhoods $V_1,...,V_k$ of the components such that for any choice of finitely many distinct completely mixed strategy profiles $\{\s^{ij}\}_{ij}, \s^{ij} \in V_i$ for each $i = 1, \ldots, k$ and numbers $r_{ij} \in \{-1,1\}$ such that $\sum_j r_{ij} = c_i$, there exists a normal-form game $\bar G$ obtained from  $G$ by adding duplicate strategies and an $\e$-perturbation $\bar G^\e$ of $\bar G$ such that the set of equilibria of $\bar G^{\e}$ is  $\{\bar \s^{ij} \}_{ij}$, where for each $i, j$: (1) $\bar \s^{ij}$ is equivalent to the profile $\s^{ij}$;  (2) the index $\bar \s^{ij}$ equals $r_{ij}$.

\end{abstract}

\title[O'Neill's Theorem for Games]{ O'Neill's Theorem for Games\\}
\author[S. Govindan]{Srihari Govindan}
\address{Department of Economics, University of Rochester, NY 14627, USA.}
\email{s.govindan@rochester.edu}
\author[R. Laraki]{Rida Laraki}
\address{Moroccan Center for Game Theory, University Mohammed VI Polytechnic (UM6P), Rabat, Morocco.}
\noindent  
\email{rida.laraki@um6p.ma}
\author[L. Pahl]{Lucas Pahl}
\address{School of Economics, University of Sheffield, Sheffield, S10 2TN, UK.}
\email{l.pahl@sheffield.ac.uk} 
\date{\today}
\thanks{We would like to thank Klaus Ritzberger and three anonymous referees for their comments. Lucas Pahl acknowledges financial support from the Hausdorff Center for Mathematics (DFG project no. 390685813).}
\maketitle

\section{Introduction}

Multiplicity of equilibria is a pervasive phenomenon in games, and refinement theory aims to reduce it by strengthening the Nash criterion through the imposition of additional requirements. For most well-known refinement concepts,  their power derives from  asking for robustness of the equilibria with respect to some type of perturbation. Perfect \cite{RS1975}, Proper \cite{RM1978} and Stable equilibria (either the Kohlberg-Mertens \cite{KM1986} or the Mertens \cite{M1989} variants) are examples of concepts that require robustness of Nash equilibria with respect to perturbations of the players' strategies. Essentiality \cite{WJ1962} and Hyperstability \cite{KM1986} are examples of concepts that require robustness to payoff perturbations. Most strategy perturbations can be encoded in payoff perturbations, which implies that concepts that require robustness to payoff perturbations are usually more stringent than the strategic ones. In this paper we focus on payoff perturbations.  

Every game induces an associated fixed-point problem, where the fixed-point map (usually called a Nash-map) or best-reply correspondence has as its fixed points the Nash equilibria of the game. Small perturbations of a game generate close-by perturbations of the fixed-point map, and therefore fixed-point theory can also inform us in the search for equilibria that are robust to payoff perturbations. The precise fixed-point theoretic tool that allows us to identify robustness of equilibria to perturbations of the fixed-point map is called the \textit{fixed-point index}. The fixed-point index is an integer number associated to each connected component of fixed points of a map and provides a characterization of robustness: a component is robust to perturbations of its fixed-point map to nearby maps if and only if its index is nonzero. 

The lesson for game theory obtained from this characterization is somewhat ambiguous: though nonzero index components are robust to payoff perturbations (cf. \cite{KR1994}), the converse is not true (see \cite{HH2002}). Moreover, a fundamental result in fixed point theory, O'Neill's Theorem (Theorem 5.2 in \cite{O1953}) implies that given any fixed point in a component, there is a nearby map for which it is the  the unique fixed point in a neighborhoord of the component. Though this result permits the interpretation that no fixed point in a connected component of fixed points is particularly distinguishable from the point of view of robustness to perturbations of the map, it yields no immediate prescription for game theory in terms of equilibrium selection: maps that approximate a Nash map may not be Nash maps of payoff-perturbed games, and therefore the same lesson cannot be translated immediately to selection of equilibria through payoff robustness.

Govindan and Wilson \cite{GW2005} proved that if we allow for duplicate strategies, then the class of fixed point maps generated by payoff perturbations is rich enough to capture the notion of essentiality in fixed point theory, i.e., of robustness to map-perturbations.\footnote{The addition of duplicate strategies is not only of technical importance to bridge the gap between the space of maps and that of payoff perturbations in the result of \cite{GW2005}, but is important in the theory of refinements as a decision-theoretic property: Kohlberg and Mertens list the property of \textit{Invariance} as a requirement a solution concept should satisfy, in the sense that a solution concept should not depend on the addition/elimination of duplicate strategies to the game, i.e., the solution should depend only on the reduced normal form.} More precisely, they proved that a Nash component $C$ of a finite game has index zero if and only if for any $\e>0$, there exists an equivalent game and an $\e$-payoff perturbation of the equivalent game with no equilibria near the equivalent component to $C$. This result provides a clear picture of the relation between robustness to map perturbations and to payoff perturbations, but is still incomplete. What about equilibrium components of nonzero index? What is the structure of the equilibria generated by payoff perturbations around positive index components? 

For general fixed-point problems, this question is answered by O'Neill's Theorem, which asserts that if $f : U \to X$ is a continuous function in a Euclidean neighborhood $U$  of a fixed point component $C$ of $f$, $r_1, \ldots, r_k$ are integers whose sum is the fixed-point index of $C$, and $x_1, \ldots, x_k$ are distinct points of $C$, then there is a map arbitrarily close to $f$ whose fixed points are $x_1, \ldots, x_k$, with the fixed point index of each $x_i$ being $r_i$. The result of O'Neill provides us with the precise topological invariant to understand which fixed-point components can be robust, namely, the index of a component: for example, if a component has an index of $+2$, one can select any finite number of points inside the component and assign to them integers such that, as long as the sum of the integers is $+2$, an arbitrarily nearby map exists with those points as the only fixed points around the component, and the indices allocated to them are precisely the integers we have chosen. 

For game-theoretic problems, the immediate question is whether a similar addition of duplicate strategies and perturbation of payoffs of equivalent games would yield like in \cite{GW2005} an analog of O'Neill's Theorem for games. This is the main result we present in this paper. In essence, we prove that one can select in each equilibrium component a finite number of equilibria and associate to each one of them an integer equal to $+1$ or $-1$, such that the sum of the numbers allocated to the selected points of the same component equals its index. For each $\e>0$, one can then construct an equivalent game---obtained by adding duplicate strategies to the original game---and an $\e$-payoff perturbation of this equivalent game with a single Nash equilibrium close to each selected point, each of them having an index precisely equal to the pre-assigned number.

Let us illustrate this result in an example. Consider the following bi-matrix game (from Kohlberg and Mertens \cite{KM1986}). It has a unique Nash component of equilibria, homeomorphic to a circle.

 
   \[
    \begin{array}{cccc}
        & L & M & R \\ 
        \cline{2-4}
        t & \multicolumn{1}{|c}{\tikz[remember picture]\node[red,circle,fill,inner sep=2pt] (nodeTL) {}; (1, 1)} & \multicolumn{1}{|c|}{(0, -1)} & \multicolumn{1}{|c|}{\tikz[remember picture]\node[red,circle,fill,inner sep=2pt] (nodeTR) {}; (-1, 1)} \\ 
        \cline{2-4}
        m & \multicolumn{1}{|c}{(-1, 0)} & \multicolumn{1}{|c|}{\tikz[remember picture]\node[red,circle,fill,inner sep=2pt] (nodeMM) {}; (0, 0)} & \multicolumn{1}{|c|}{\tikz[remember picture]\node[red,circle,fill,inner sep=2pt] (nodeMR) {}; (-1, 0)} \\ 
        \cline{2-4}
        b & \multicolumn{1}{|c}{\tikz[remember picture]\node[red,circle,fill,inner sep=2pt] (nodeBL) {}; (1, -1)} & \multicolumn{1}{|c|}{\tikz[remember picture]\node[red,circle,fill,inner sep=2pt] (nodeBM) {}; (0, -1)} & \multicolumn{1}{|c|}{(-2, -2)} \\ 
        \cline{2-4}
    \end{array}
    \]

    \begin{tikzpicture}[overlay, remember picture]
        \draw[red, ultra thick] (nodeTL) -- (nodeBL);  
        \draw[red, ultra thick] (nodeTL) -- (nodeTR);  
        \draw[red, ultra thick] (nodeTR) -- (nodeMR);  
        \draw[red, ultra thick] (nodeMR) -- (nodeMM);  
        \draw[red, ultra thick] (nodeMM) -- (nodeBM);  
        \draw[red, ultra thick] (nodeBM) -- (nodeBL);  
    \end{tikzpicture}

Say we would like to make $(\frac{1}{2}t+\frac{1}{2}b, L)$ the unique equilibrium (and so necessarily with index $+1$) of a perturbed equivalent game, we duplicate strategy $L$ and perturb the payoffs as follows. 
\smallskip
$$
\begin{array}{ccccc}
	& L& L' & M & R \\
	\cline{2-5}
	t & \multicolumn{1}{|c}{(1+\varepsilon, 1)} &\multicolumn{1}{|c}{(1, 1+\varepsilon)} & \multicolumn{1}{|c|}{(0+\varepsilon,-1)}& \multicolumn{1}{|c|}{(-1+\varepsilon, 1)} \\
	\cline{2-5}
	m & \multicolumn{1}{|c}{(-1, 0)} &\multicolumn{1}{|c}{(-1, +\varepsilon)} & \multicolumn{1}{|c|}{(0 ,0)}& \multicolumn{1}{|c|}{(-1, 0)} \\
	\cline{2-5}
	b & \multicolumn{1}{|c}{(1, -1+\varepsilon)} &\multicolumn{1}{|c}{(1+\varepsilon, -1)} & \multicolumn{1}{|c|}{(0,-1)}& \multicolumn{1}{|c|}{(-2, -2)} \\
	\cline{2-5}
\end{array}
$$
\smallskip
As $m$, $M$ and $R$ are strictly dominated in the above game, eliminating them leads to a two-by-two game where the unique Nash equilibrium is $(\frac{1}{2}t+\frac{1}{2}b;\frac{1}{2}L+\frac{1}{2}L')$ (necessarily with index +1) which is equivalent in the original game to  $\sigma=(\frac{1}{2}t+\frac{1}{2}b, L)$. 

Consider now the three Nash equilibria $(t,L)$ and $(b,L)$ and $\sigma=(\frac{1}{2}t+\frac{1}{2}b, L)$ and associate to the first two the integer $+1$ and to the last the integer $-1$. We will show how to add to the original game the same duplicate strategy $L'$ as above and perturb the payoffs so as to generate exactly three equilibria in a perturbed game. Consider the following perturbation: 
\medskip
$$
\begin{array}{ccccc}
	& L& L' & M & R \\
	\cline{2-5}
	t & \multicolumn{1}{|c}{(1+\varepsilon, 1+\varepsilon)} &\multicolumn{1}{|c}{(1, 1)} & \multicolumn{1}{|c|}{(0+\varepsilon,-1)}& \multicolumn{1}{|c|}{(-1+\varepsilon, 1)} \\
	\cline{2-5}
	m & \multicolumn{1}{|c}{(-1, 0)} &\multicolumn{1}{|c}{(-1, +\varepsilon)} & \multicolumn{1}{|c|}{(0\,0)}& \multicolumn{1}{|c|}{(-1, 0)} \\
	\cline{2-5}
	b & \multicolumn{1}{|c}{(1, -1)} &\multicolumn{1}{|c}{(1+\varepsilon, -1+\varepsilon)} & \multicolumn{1}{|c|}{(0,-1)}& \multicolumn{1}{|c|}{(-2, -2)} \\
	\cline{2-5}
\end{array}
$$

Strategies $m$, $M$ and $R$ are strictly dominated. Eliminating them leads to a battle of the sexes where only three Nash equilibria are present: the two strict equilibria $(t,L)$ and $(b,L')$ (hence with indices $+1$ each) and a completely mixed equilibrium $(\frac{1}{2}t+\frac{1}{2}b;\frac{1}{2}L+\frac{1}{2}L')$ with index $-1$. But since $L'$ was a duplicate of $L$, we obtain as Nash equilibria $(t,L)$ and $(b,L)$ and $(\frac{1}{2}t+\frac{1}{2}b, L)$, and their indices correspond to the integers we fixed previously. Hence, perturbations to nearby equivalent games can single out $\sigma$ as an isolated equilibrium with different indices: in the first perturbation we obtained it with index $+1$ and in the second with index $-1$. 

Our proof of the extension of O'Neill's Theorem to games is inspired by the one in Govindan Wilson \cite{GW2005}. Though the steps of the proof are similar, the details are very different and pose different technical challenges. It requires more preliminary work in the realm of fixed point theory in order to produce a suitable approximation of the best-reply correspondence (which is achieved with the help of a fundamental result in obstruction theory, namely, the Hopf Extension Theorem). Then we construct a game that has essentially that approximation as its best-reply correspondence (which will require tools from the theory of triangulations). 

The paper is organized as follows. Section \ref{sec definitions} sets up the problem and states the main result. Section \ref{sec key ideas} gives a detailed outline of the proof and includes a table of notation that is used in the proof. Section \ref{sec proof} contains the proof of the theorem, which is presented in a sequence of steps.   

\section{Preliminary Definitions and Statement of the Theorem}\label{sec definitions}

We start by introducing some notational conventions. A finite normal-form game is a tuple $G = (\N, (S_n)_{n \in \N}, (G_n)_{n \in \N})$, where $\N = \{1,...,N\}$ is the set of players, $S_n$ is a finite set of pure strategies of player $n$, $S \equiv \times_{n \in \N}S_n$, and $G_n: S \to \Re$ is the payoff function of player $n$. We denote the set of mixed strategies of player $n$ by $\S_n \equiv \D(S_n)$, and the set of mixed-strategy profiles is denoted $\S \equiv \prod_n \S_n$. For each $n$, we continue to use $G_n$ to denote the extension of $n$'s payoff function to $\S$.  The best-reply correspondence of $G$ is denoted $\BR$.

For $V_n \subseteq \S_n$,  $\partial V_n$ denotes the boundary of $V_n$ in the affine space generated by $\S_n$; $\partial_{\S_n} V_n$ refers to the boundary of $V_n$ in $\S_n$. Similarly, $\partial V$ and $\partial_{\S} V$ denote, resp.,  the boundaries of $V \subseteq \S$ in the affine space generated by $\S$ and in $\S$. Finally, $\Vert \cdot \Vert$ denotes the $\ell_\infty$-norm on Euclidean spaces.

A \textit{polyhedral subdivision} or \textit{polyhedral complex} $\P$ of a polytope $X \subset \Re^m$ is a collection of polytopes in $\Re^m$ that satisfies three conditions: first, the union of the polytopes in $\P$ is equal to $X$; second, every face of a polyhedron in the collection belongs to the collection; third, the intersection of any two polytopes in $\P$ is a face of each (possibly empty). Given a polyhedral subdivision $\P$ of $X$, a typical member of $\P$ is called a \textit{cell}. The cells of dimension $0$ are called \textit{vertices}.  When the polyhedral subdivision is made of simplices, it is called a \textit{triangulation}. The \textit{closed star of a vertex $v$ of a triangulation $\T_n$} is the union of simplices that have $v$ as a vertex; and the \textit{simplicial neighborhood} of the closed star of $v$ is the union of all simplices that intersect the closed star. 

\begin{example} 
 Figure \ref{triang1} provides an example of a triangulation of a square. This triangulation has 16 vertices (labeled from $a$ to $p$), 33 1-simplices, and 18 2-simplices. The space of this triangulation is the union of all the simplices, therefore, the square itself. The closed star of vertex $f$ in the triangulation is depicted in red in Figure \ref{star} and the simplicial neighborhood of this closed star is depicted in Figure \ref{simplicial}.

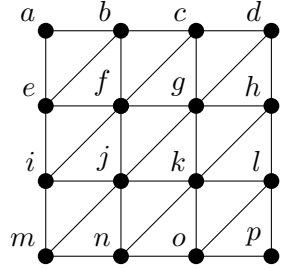
\begin{figure}
\caption{Triangulation of a square}\label{triang1}

\begin{tikzpicture}
  \draw
    (0, 0) grid[step=1cm] (3, 3)
    (0, 2) -- (1, 3)
    (0, 1) -- (2, 3)
    (0, 0) -- (3, 3)
    (1, 0) -- (3, 2)
    (2, 0) -- (3, 1)
  ;

  \fill[radius=3pt]
    \foreach \x in {0, ..., 3} {
      \foreach \y in {0, ..., 3} {
        (\x, \y) circle[]
      }
    }
  ;
 
  \fill[black] (1, 2) circle[radius=3pt];

  \path[above left]
    \foreach \y in {3} {
      \foreach[count=\x] \v in {a, b, c, d} {
        (\x-1, \y) node {$\v$}
      }
    }

\foreach \y in {0} {
      \foreach[count=\x] \v in {m, n, o, p} {
        (\x-1, \y) node {$\v$}
      }
    }
    \foreach \p/\v in {
      {0, 2}/e,
      {1, 2}/f,
      {2, 2}/g,
      {3, 2}/h,
      {0, 1}/i,
      {1, 1}/j,
      {2, 1}/k,
      {3, 1}/l%
    } {
      (\p) node {$\v$}
    }
  ;

\end{tikzpicture}

\end{figure}


\begin{figure} 
\caption{The closed star of vertex $f$}\label{star}
\begin{tikzpicture}
  \draw
    (0, 0) grid[step=1cm] (3, 3)
    (0, 2) -- (1, 3)
    (0, 1) -- (2, 3)
    (0, 0) -- (3, 3)
    (1, 0) -- (3, 2)
    (2, 0) -- (3, 1)
  ;
  
  \fill[red!20]
    (1, 2) -- (0, 2) -- (1, 3) -- cycle 
    (1, 2) -- (2, 2) -- (2, 3) -- cycle 
    (1, 2) -- (1, 1) -- (2, 2) -- cycle 
    (1, 2) -- (0, 1) -- (0, 2) -- cycle 
    (1,2) -- (1,1) -- (0,1) -- cycle 
	(1,2) -- (1,3) -- (2,3) -- cycle 
  ;

  \draw[thick, red]
    (1, 2) -- (0, 2)
    (1, 2) -- (2, 2)
    (1, 2) -- (0, 1)
    (1, 2) -- (1, 1)
    (1, 2) -- (2, 3)
	(1,2) -- (1,3)
	(1,1) -- (2,2)
	(2,2) -- (2,3)
	(0,2) -- (1,3)
	(1,3) -- (2,3)
	(0,2) -- (0,1)
	(0,1) -- (1,1)
  ;

  \fill[radius=3pt]
    \foreach \x in {0, ..., 3} {
      \foreach \y in {0, ..., 3} {
        (\x, \y) circle[]
      }
    }
  ;
 
  \fill[red] 
			(1, 2) circle[radius=3pt]
			(2,2) circle[radius=3pt]
			(2,3) circle[radius=3pt]
			(1,3) circle[radius=3pt]
			(0,2) circle[radius=3pt]
			(0,1) circle[radius=3pt]
			(1,1) circle[radius=3pt]

;

  \path[above left]
    \foreach \y in {3} {
      \foreach[count=\x] \v in {a, b, c, d} {
        (\x-1, \y) node {$\v$}
      }
    }

\foreach \y in {0} {
      \foreach[count=\x] \v in {m, n, o, p} {
        (\x-1, \y) node {$\v$}
      }
    }
    \foreach \p/\v in {
      {0, 2}/e,
      {1, 2}/f,
      {2, 2}/g,
      {3, 2}/h,
      {0, 1}/i,
      {1, 1}/j,
      {2, 1}/k,
      {3, 1}/l%
    } {
      (\p) node {$\v$}
    }
  ;

\end{tikzpicture}

\end{figure}
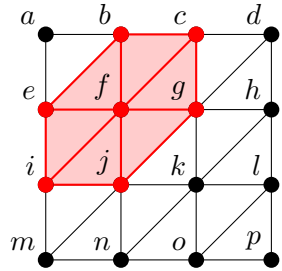 

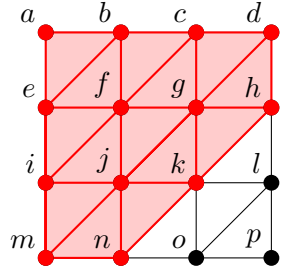
\begin{figure} 
\caption{The simplicial neighborhood of the closed star of $f$}\label{simplicial}
\begin{tikzpicture}
  \draw
    (0, 0) grid[step=1cm] (3, 3)
    (0, 2) -- (1, 3)
    (0, 1) -- (2, 3)
    (0, 0) -- (3, 3)
    (1, 0) -- (3, 2)
    (2, 0) -- (3, 1)
  ;
  
  \fill[red!20]
    (1, 2) -- (0, 2) -- (1, 3) -- cycle 
    (1, 2) -- (2, 2) -- (2, 3) -- cycle 
    (1, 2) -- (1, 1) -- (2, 2) -- cycle 
    (1, 2) -- (0, 1) -- (0, 2) -- cycle 
    (1,2) -- (1,1) -- (0,1) -- cycle 
	(1,2) -- (1,3) -- (2,3) -- cycle 
	(0,0) -- (0,1) -- (1,1) -- cycle
	(0,0) -- (1,0) -- (1,1) -- cycle
	(1,0) -- (1,1) -- (2,1) -- cycle
	(1,1) -- (2,1) -- (2,2) -- cycle
	(2,1) -- (2,2) -- (3,2) -- cycle
	(2,2) -- (3,2) -- (3,3) -- cycle
	(2,2) -- (2,3) -- (3,3) -- cycle
	(0,2) -- (0,3) -- (1,3) -- cycle
  ;

  \draw[thick, red]
    (1, 2) -- (0, 2)
    (1, 2) -- (2, 2)
    (1, 2) -- (0, 1)
    (1, 2) -- (1, 1)
    (1, 2) -- (2, 3)
	(1,2) -- (1,3)
	(1,1) -- (2,2)
	(2,2) -- (2,3)
	(0,2) -- (1,3)
	(1,3) -- (2,3)
	(0,2) -- (0,1)
	(0,1) -- (1,1)
	(0,0) -- (0,1)
	(0,0) -- (0,2)
	(0,0) -- (0,3)
	(0,0) -- (1,1)
	(1,1) -- (2,2)
	(2,2) -- (3,3)
	(1,1) -- (2,1)
	(0,0) -- (1,0)
	(1,0) -- (1,1)
	(1,0) -- (2,1)
	(2,1) -- (2,2)
	(2,1) -- (3,2)
	(2,2) -- (3,2)
	(3,2) -- (3,3)
	(2,3) -- (3,3)
	(0,3) -- (1,3)

  ;

  \fill[radius=3pt]
    \foreach \x in {0, ..., 3} {
      \foreach \y in {0, ..., 3} {
        (\x, \y) circle[]
      }
    }
  ;
 
  \fill[red] 

			(1, 2) circle[radius=3pt]
			(2,2) circle[radius=3pt]
			(2,3) circle[radius=3pt]
			(1,3) circle[radius=3pt]
			(0,2) circle[radius=3pt]
			(0,1) circle[radius=3pt]
			(1,1) circle[radius=3pt]
			(0,3) circle[radius=3pt]
			(0,0) circle[radius=3pt]
			(1,0) circle[radius=3pt]
			(2,1) circle[radius=3pt]
			(3,2) circle[radius=3pt]
			(3,3) circle[radius=3pt]
			
;  

  \path[above left]
    \foreach \y in {3} {
      \foreach[count=\x] \v in {a, b, c, d} {
        (\x-1, \y) node {$\v$}
      }
    }

\foreach \y in {0} {
      \foreach[count=\x] \v in {m, n, o, p} {
        (\x-1, \y) node {$\v$}
      }
    }
    \foreach \p/\v in {
      {0, 2}/e,
      {1, 2}/f,
      {2, 2}/g,
      {3, 2}/h,
      {0, 1}/i,
      {1, 1}/j,
      {2, 1}/k,
      {3, 1}/l%
    } {
      (\p) node {$\v$}
    }
  ;

\end{tikzpicture}
\end{figure}

\end{example}

The \textit{face complex} of a polytope $X \subset \Re^m$ is the polyhedral subdivision $\P$ of $X$ given by the collection of its faces ($X$ and the empty set are also considered faces, so they also belong to the face complex). If $\P$ is any polyhedral subdivision of $X$, then a \textit{subcomplex} of $\P$ is a subset of $\P$ that is itself a polyhedral subdivision. The \textit{space of a complex} is the union of the cells that belong to it.

\begin{example}\label{examplefacecomplex}

In Figure \ref{fcom} we have a 2-simplex ABC. The face complex of this 2-simplex is given by the empty face,  vertices $A,B,C$, 1-simplices $[A,B], [B,C], [A,C]$, and the 2-simplex $[A,B,C]$. We highlighted in red the space of the subcomplex of the face complex composed by vertices $A,B,C$ and 1-simplices $[A,B]$ and $[B,C]$. 

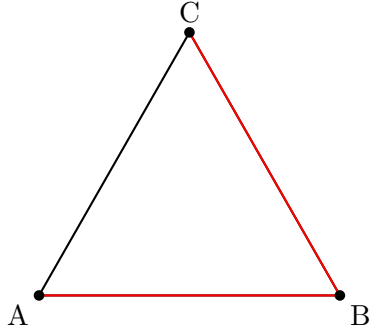
\begin{figure}
\caption{Face complex and subcomplex}\label{fcom}

\begin{tikzpicture}

\coordinate (A) at (0, 0);
\coordinate (B) at (4, 0);
\coordinate (C) at (2, 3.5);  

\draw[thick] (A) -- (C) -- (B) -- cycle;

\draw[thick, red] (A) -- (B);
\draw[thick, red] (B) -- (C);

\fill[black] (A) circle (2pt) node[anchor=north east] {A};
\fill[black] (B) circle (2pt) node[anchor=north west] {B};
\fill[black] (C) circle (2pt) node[anchor=south] {C};

\end{tikzpicture}

\end{figure}
\end{example}

Given a polyhedral subdivison $\T_n$ of space $X$ and $x \in X$, the \textit{carrier of $x$ in subdivision $\T_n$} is the unique cell of $\T_n$ that contains $x$ in its interior (relative to the affine space generated by the cell). Whenever the underlying polyhedral subdivision is understood, we shall omit reference to $\T_n$. Figure \ref{carrier}  illustrates the carrier of point $x$ in a triangulation of a 2-simplex: the carrier of point $x$ is the 2-simplex $AGC$.

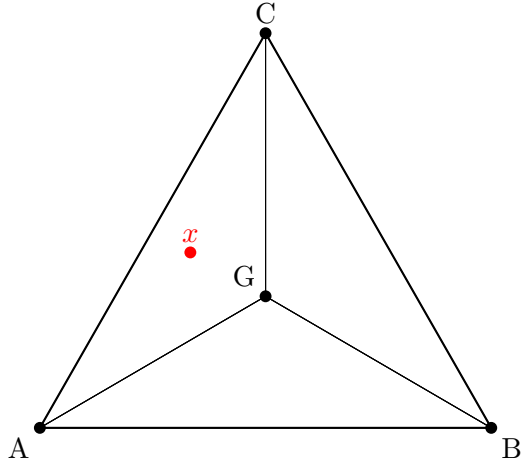
\begin{figure} 
\caption{Carrier of point in $x$ in a triangulation}\label{carrier}

\begin{tikzpicture}[scale=1.5]

\coordinate (A) at (0, 0);
\coordinate (B) at (4, 0);
\coordinate (C) at (2, 3.5);  

\coordinate (G) at (barycentric cs:A=1,B=1,C=1);  

\draw[thick] (A) -- (B) -- (C) -- cycle;

\draw (G) -- (A);
\draw (G) -- (B);
\draw (G) -- (C);

\draw (A) -- (G) -- (B) -- cycle;
\draw (B) -- (G) -- (C) -- cycle;
\draw (C) -- (G) -- (A) -- cycle;

\coordinate (F1) at (barycentric cs:A=1,G=1,C=1); 

\fill[black] (A) circle (1.5pt) node[anchor=north east] {A};
\fill[black] (B) circle (1.5pt) node[anchor=north west] {B};
\fill[black] (C) circle (1.5pt) node[anchor=south] {C};
\fill[black] (G) circle (1.5pt) node[anchor=south east] {G};

\fill[red] (F1) circle (1.5pt) node[anchor=south] {$x$};

\end{tikzpicture}
\end{figure}

For technical reasons, our proof will require considering the notion of a game defined over slightly more general strategy sets than simplices for each player (cf. Pahl \cite{LP2020}). 

\begin{definition} A \textit{polytope-form game} is a tuple $G = (\N, (P_n)_{n \in \N}, (V_n)_{n \in \N})$, where $\N = \{1,...,N\}$ is the set of players, $P_n$ is a polytope in $\Re^{m_n}$ (the set of strategies of player $n$) and $G_n: \prod_{n \in \N}P_n \to \Re$ is the payoff function of player $n$, which is affine in each coordinate $p_j \in P_j, j=1,...,N$.\footnote{The original term to denote this class of games is \textit{strategic-form games} and was coined in \cite{M2004}. In order to avoid confusion with the current terminology of ``strategic-form games'', we opted for the ``polytope-form game'' designation.} \end{definition}

Every finite game in mixed strategies can be viewed as a game in polytope-form since strategy sets are simplices hence polytopes and payoff functions are multi-linear.  Going the other way, given a polytope-form game, we can define a normal-form game where the set of pure strategies of each player is the set of vertices of $P_n$; these two are ``equivalent'' in the sense to be defined now.

\begin{definition}\label{equivalence} Two  games $G = (\N, (P_n)_{n \in \N}, (G_n)_{n \in \N})$ and $G' = (\N, (P'_n)_{n \in \N}, (G'_n)_{n \in \N})$ in polytope form  are \textit{equivalent} if there exist a polytope-form game $\tilde G = (\N, (\tilde P_n)_{n \in \N}, (\tilde G_n)_{n \in \N})$, affine and surjective maps $\phi_n: P_n \to \tilde P_n$ and $\phi'_n: P'_n \to \tilde P_n$, such that $G_n(p) = \tilde G_n(\phi (p))$ and $G'_n(p') = \tilde G_n(\phi'(p'))$, with $\phi \equiv \times_n \phi_n, \phi' \equiv \times_n \phi'_n$. 
\end{definition}

Let $G$ and $\tilde G$ be as in Definition \ref{equivalence}. We say $p_n \in P_n$ \textit{projects} to $\tilde p_n \in \tilde P_n$, if $\phi_n(p_n) = \tilde p_n$; similarly, we say $p \in P$ projects to $\tilde p \in \tilde P$, if $\phi (p) = \tilde p$. We say that a set $A \subseteq P$ projects to  $\tilde A \subseteq \tilde P$ if $\phi(A) = \tilde A$. Given a normal-form game $G$, a \textit{duplicate strategy $s_n$ of player $n$} is a pure strategy that projects to a (possibly mixed) strategy of player $n$ that is distinct from $s_n$. A normal-form game $\bar G$ obtained by adding duplicate strategies from $G$ is a finite game where each pure strategy is either a pure strategy of $G$ or is a duplicate of a strategy of $G$ (pure or mixed).

\begin{example} Assume player $n$ has three pure strategies $\{A,B,C\}$ in a finite normal-form game $G$. His mixed strategy set is therefore a 2-simplex as depicted in Figure \ref{equivalent}. We now introduce the duplicate strategy $b$ for player $n$, which corresponds to the barycenter of the 2-simplex. The newly defined pure strategy set of the player is then $\{A,B,C,b\}$, and the mixed strategy set is a tetrahedron. Payoffs for the player in the equivalent game $\bar{G}$ where the strategy set of the player is a tetrahedron are defined by projecting $b$ to the 2-simplex, that is: $\bar{G}_n(b, s_{-n}) = \frac{1}{3}G_n(A, s_{-n}) + \frac{1}{3}G_n(B,s_{-n}) + \frac{1}{3}G_n(C, s_{-n})$.

\begin{figure}
\caption{Strategy set of player $n$ and equivalent strategy set}\label{equivalent}
\begin{tikzpicture}[scale=1.0]

\coordinate (A) at (0, 0);
\coordinate (B) at (4, 0);
\coordinate (C) at (2, 3.5);  

\coordinate (G) at (barycentric cs:A=1,B=1,C=1);  

\draw[thick] (A) -- (B) -- (C) -- cycle;

\coordinate (F1) at (barycentric cs:A=1,G=1,C=1); 

\fill[black] (A) circle (1.5pt) node[anchor=north east] {A};
\fill[black] (B) circle (1.5pt) node[anchor=north west] {B};
\fill[black] (C) circle (1.5pt) node[anchor=south] {C};
\fill[black] (G) circle (1.5pt) node[anchor=south east] {b};


\end{tikzpicture}

\medskip
\begin{tikzpicture}[scale=2.5]

\coordinate (A) at (0, 0);
\coordinate (C) at (1.2, 0); 
\coordinate (G) at (0.5, 1.5); 
\coordinate (B) at (1.8, 0.5); 

\coordinate (F1) at ($(A)!1/3!(C)!1/3!(G)$); 

\draw[thick] (A) -- (B) -- (C) -- cycle;

\draw[thick] (G) -- (A);
\draw[thick] (G) -- (B);
\draw[thick] (G) -- (C);

\draw[thick] (B) -- (G);  

\filldraw[black] (A) circle (1pt) node[anchor=north east] {A};
\filldraw[black] (B) circle (1pt) node[anchor=west] {B};
\filldraw[black] (C) circle (1pt) node[anchor=south] {C};
\filldraw[black] (G) circle (1pt) node[anchor=south] {b};


\end{tikzpicture}

\end{figure}
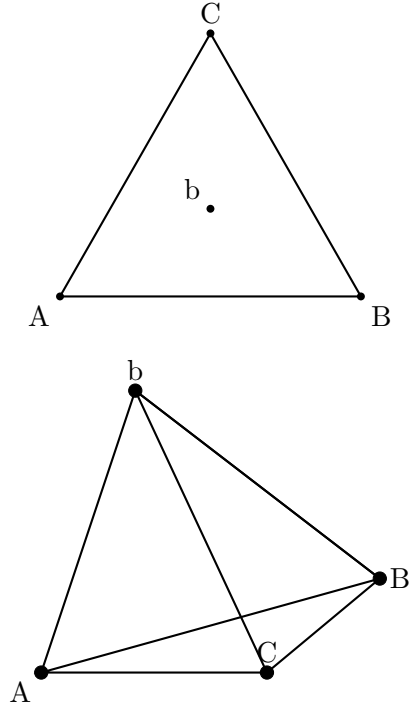

\end{example}

\subsection{Fixed Point Index}
We provide a brief review of the concept of fixed point index here. For a comprehensive and axiomatic treatment of the subject we refer the reader to \cite{AM2018}. Let $X$ be a convex set and $A$ a compact subset of $X$. Say that a correspondence $\varphi: A \twoheadrightarrow X$ is well-behaved if it is nonempty, compact, convex-valued and upper semicontinuous. Suppose $\varphi$ is a well-behaved correspondence without  fixed points on the boundary of $A$ (in $X$).  The index of $\varphi$ (over $A$) is an integer that serves as an algebraic count of the number of fixed points of $\varphi$ in $A$. In the case where $A$ is the closure of an open subset of an Euclidean space, and $\varphi$ is a function, the index of $\varphi$ can be computed as the degree of the displacement $d_\varphi$ of $\varphi$, which is defined as $d_\varphi \equiv \text{Id} - \varphi$.

We invoke three important properties of the fixed point index.  The first property is continuity: there exists $\e > 0$ such that for all well-behaved correspondences $\varphi': A \twoheadrightarrow X$ whose graph is in the $\e$-neighborhood of that of $\varphi$, the index of $\varphi'$ is the same as that of $\varphi$.  Since a well-behaved correspondence $\varphi$ can be approximated by functions whose graphs are contained in arbitrarily small neighborhoods of the graph of $\varphi$, the index of $\varphi$ can be computed by approximating it by a function.

The second property, which is related to continuity, is homotopy.  Given two well-behaved correspondences $\varphi, \varphi': A \twoheadrightarrow X$, if $\l\varphi + (1-\l)\varphi'$ has no fixed points on the boundary of $A$ for any $\l \in [0, 1]$, then the index of $\varphi$ and $\varphi'$ are the same.

The third property we need is the multiplicative property.  If $\varphi_1: A_1 \twoheadrightarrow X_1$ and $\varphi_2: A_2 \twoheadrightarrow X_2$ are two-well behaved correspondences without fixed points on the boundaries of $A_1$ and $A_2$, the index of $\varphi_1 \times \varphi_2$ is the product of the indices of $\varphi_1$ and $\varphi_2$.

Suppose $\varphi: X \twoheadrightarrow X$ is a well-behaved correspondence and $X$ is compact and convex.  Let $C$ be a component of fixed points of $\varphi$.  Let $A$ be a closed neighborhood of $C$ that contains no other fixed points.  The index of $C$ is defined to be the index of $\varphi$ over $A$.   The index of $C$ is independent of the choice of the neighborhood, so long as the neighborhood does not contain any other fixed points.  

Given a game $G$ in polytope-form, let $\varphi: P \twoheadrightarrow P$ be the best-reply correspondence, i.e., $\varphi(p)$ is the set of $q \in P$ such that for each $n$, $q_n$ is a best reply to $p$. $\varphi$ is a well-behaved correspondence.  The game $G$ has finitely many components of Nash equilibria, obtained as the set of fixed points of $\varphi$.  For each component $C$, we can then assign an index.  The index of Nash equilibria can also be computed using any Nash map, i.e., a function $f$ that is jointly continuous in payoffs and strategies and whose fixed points for any game are the Nash equilibria (cf. \cite{DG2000}). The sum of indices of each equilibrium component of any game is $+1$. (cf. \cite{AM2018}).

\subsection{Statement of the Theorem} Fix a finite normal-form game $G$.  Let $C_1, \ldots, C_k$ be the components of its equilibria.  Let the index of each $C_i$ be $c_i$.  

\begin{theorem}\label{main thm}
	For each $\e>0$,  there exist  pairwise disjoint sets $V_1, \ldots, V_k \subseteq \S$, $V_i$ a neighborhood of $C_i$ for each $i$, such that for any choice of finitely many distinct points $\{\s^{ij}\}_{ij}, \s^{ij} \in V_i \backslash \partial V_i$ and numbers $r_{ij} \in \{-1,1\}$ such that $\sum_j r_{ij} = c_i$, there exists a  normal-form game $\hat G$ that is equivalent to $G$ and a $\e$-perturbation $\hat G^\e$ of $\hat G$ such that: 
	\begin{enumerate} 
		\item The set of equilibria of $\hat G^{\e}$ equals $\{\hat \s^{ij} \}_{ij}$, where $\hat \s^{ij}$ projects to $\s^{ij}$, for each $i,j$. 
		\item The index of $\hat \s^{ij}$ equals $r_{ij}$.
	\end{enumerate}
\end{theorem}

The following corollary follows immediately from Theorem \ref{main thm} and distinguishes components with a positive index from the others. Below, we say that a game is \textit{generic} if it possesses finitely many equilibria, each with index $+1$ or $-1$. It is known that the set of payoffs of a game for which this property is violated has a strictly lower dimension than the whole payoff space (cf. \cite{KR1994}). 

\begin{corollary}\label{corollary}
	A set $C \subseteq \S$ of strategy profiles is a component of equilibria of $G$ with a positive index iff it is a minimal set (under set inclusion) with respect to the following property: for each $\e > 0$, there exists $\d > 0$ such that for each equivalent game $\bar G$ and each generic $\bar G^\d$ that is an $\d$-perturbation of $\bar G$, there is a $+1$ equilibrium whose projection to $\S$ is within $\e$ of $C$.   
\end{corollary}

Corollary \ref{corollary} shows that $+1$ equilibria always arise in generic perturbations of payoffs of a finite game, locally around a positive index component. In contrast, when negative index components are considered, there are always generic payoff perturbations for which all the equilibria arising locally have $-1$ index. The relevance of this result is that $+1$ equilibria possess a number of interesting properties their $-1$ counterparts do not: \cite{H2000} shows that the $+1$ equilibria of a generic game are the only ones that are dynamically stable for at least one Nash dynamics.  \cite{AM2016} argues for the selection of +1 index equilibria, on both theoretical as well as experimental grounds. \cite{HP2023} uses the $+1$ index selection criterion to refine equilibria in signaling games, and in \cite{GLP2023}, we show that they are the only equilibria that can be made unique in a larger game obtained by adding inferior replies to them.

\section{A Guide to the Proof of Theorem \ref{main thm}}\label{sec key ideas} 
This section provides a rough sketch of the proof of Theorem \ref{main thm}, which is presented formally in the next section. The main objective here is to highlight the key ideas involved in modifying O'Neill's theorem to construct a map that is game-theoretically meaningful, i.e., whose fixed points are Nash equilibria of a finite game. 

We work with the best-reply correspondence, rather than a Nash map.  The main reason is that for any approximation of the best-reply correspondence by a continuous function $f$, we have immediate information on the degree of suboptimality of $f(\s)$ against $\s$. More precisely, fix $\e > 0$ as in the statement of Theorem \ref{main thm}. Let $\BR_\e$ be the $\e$-best reply correspondence, i.e., $\BR_\e(\s)$ is the set of all profiles $\t$ such that for $n$, $\t_n$ yields a payoff against $\s$ that is strictly within $\e$ of the best payoff achievable against $\s$. The graph of the $\e$-best-reply correspondence, denoted $\Graph(\BR_\e)$, is a neighborhood of the graph of $\BR$.   For any continuous function $f$ whose graph is contained in $\Graph(\BR_\e)$, $f(\s)$ is an $\e$-best-reply to $\s$.  

The sets $V_i$ identified in the statement of the theorem are neighborhoods of the components that are pairwise disjoint and such that for any $\s$ that belongs to some $V_i$, $(\s, \s) \in \Graph(\BR_\e)$.  Now we fix points $\s^{ij}$ as in Theorem \ref{main thm}.  Our first task, carried out in Step 2, is to obtain a version of O'Neill's result by constructing an approximation of $\BR$ that has the $\s^{ij}$'s as its fixed points and with the assigned indices.  To see why we cannot directly appeal to O'Neill's result, let us recapitulate what he does. O'Neill considers the case of a map $f: A \to X$ where the set of fixed points is a connected set $C$ and $A$ is a Euclidean neighborhood of $C$.  Given finitely many points $x_1, \ldots, x_k$ and integers $r_i$ adding up to the index of $C$, O'Neill constructs a  nearby map $f'$ with exactly the $x_i$'s as the fixed points and with the given indices $r_i$, by first  perturbing the displacement map $d$ of $f$ to a map $d'$ that is close by, and then he constructs $f'$ as the map $x- d'(x)$.  Thus, compared to our set up, there are two principal differences.  Firstly, O'Neill is working with a function and not a correspondence.  Secondly, and more importantly, it is crucial to his proof that $A$ be a Euclidean neighborhood of $C$, as the following example illustrates, and this is a property that we cannot assume.

\begin{example}
Consider a map $f$ from a 2-simplex to itself. The displacement map of $f$ is defined as $d(x) = x - f(x)$. Figure \ref{interior} shows the case where a component of fixed points in red lies in the interior of the simplex, inside the Euclidean neighborhood $V$ in pale blue. As in O'Neill's proof, we would like to obtain a map $f'$ close to $f$, without any fixed points in $V$ and such that it is equal to $f$ outside of $V$. For that, the Hopf Extension Theorem is applied to the \textit{displacement map} $d$ in order to obtain a new map $d'$ which is equal to $d$ outside of the neighborhood, new and without fixed points inside of the neighborhood. After obtaining $d'$, the map defined by $x - d'(x)$ cannot be guaranteed to be inside the $2$-simplex. However, it is possible, with the help of Urysohn's Lemma, to scale $d'(x)$ by a strictly positive function $\l(\cdot)$ so as to have $f'(x) = x - \l(x) d'(x)$ inside the simplex and close to $f$, as Figure \ref{interior} indicates. 

\begin{figure}
\centering
\caption{The displacement of point $x$ under the extended map $d'$ and map $f'$}\label{interior}
\begin{tikzpicture}[scale=4]

\coordinate (A) at (0, 0);
\coordinate (B) at (1, 0);
\coordinate (C) at (0.5, 0.866);  

\coordinate (D) at (1,0.5);
\coordinate (E) at (0.48, 0.35);

\draw[thick] (A) -- (B) -- (C) -- cycle;

\filldraw[black] (A) circle (0.5pt) node[anchor=north east] {A};
\filldraw[black] (B) circle (0.5pt) node[anchor=north west] {B};
\filldraw[black] (C) circle (0.5pt) node[anchor=south] {C};

\coordinate (G) at (barycentric cs:A=1,B=1,C=1);

\filldraw[blue!20, opacity=0.7] (G) circle (0.2);  

\coordinate (P1) at (0.47, 0.43);  
\coordinate (P2) at (0.5, 0.2);   

\draw[red, thick] (P1) -- (P2);

\draw[blue, dashed] (D) -- (E);

\coordinate (P3) at (0.7, 0.415); 
\filldraw[blue] (P3) circle (0.3pt) node[anchor=north west] {$f'(x)$};

\filldraw[red] (E) circle (0.3pt) node[anchor=south] {$x$};
\filldraw[blue] (D) circle (0.3pt) node[anchor=south] {$x - d'(x)$};

\end{tikzpicture}
\end{figure}
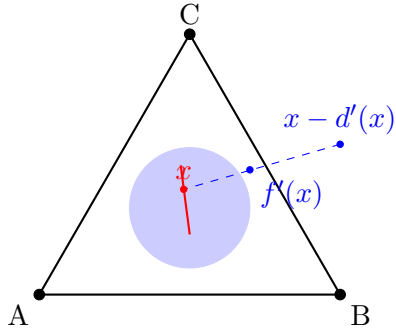

However, if the component is in the boundary of the simplex like in Figure \ref{boundary}, the reasoning employed for the interior case does not apply, i.e., O'Neill's construction does not hold. As the figure suggests, in the case $ x$ is mapped to the point $x - d'(x)$ as indicated, no small positive scaling of $d'(x)$ by $\lambda>0$ guarantees that $ x - \l d'(x)$ is inside of the simplex.

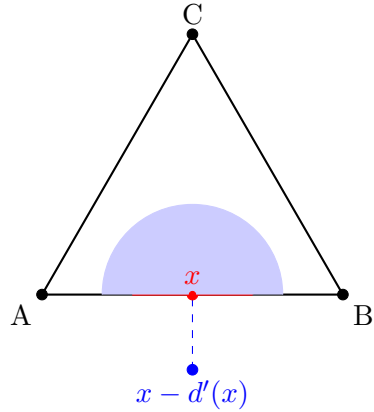
\begin{figure} 
\caption{A component of fixed points in the boundary}\label{boundary}
\begin{tikzpicture}[scale=4]

\coordinate (A) at (0, 0);
\coordinate (B) at (1, 0);
\coordinate (C) at (0.5, 0.866);  

\coordinate (D) at (0.5, -0.25);  
\coordinate (E) at (0.5,0);

\draw[thick] (A) -- (B) -- (C) -- cycle;

\filldraw[black] (A) circle (0.5pt) node[anchor=north east] {A};
\filldraw[black] (B) circle (0.5pt) node[anchor=north west] {B};
\filldraw[black] (C) circle (0.5pt) node[anchor=south] {C};
\filldraw[blue] (D) circle (0.5pt) node[anchor=north] {$x- d'(x)$};
\coordinate (P1) at (0.3, 0);  
\coordinate (P2) at (0.7, 0);  

\draw[red, thick] (P1) -- (P2);
\draw[blue, dashed] (D) -- (E);

\filldraw[red] (P1) circle (0pt);
\filldraw[red] (P2) circle (0pt);
\filldraw[red] (E) circle (0.3pt) node[anchor=south] {$x$};
\filldraw[red] (E) circle (0.5pt);

\coordinate (Mid) at ($(P1)!0.5!(P2)$);  

\def\radius{0.3}

\begin{scope}
\clip (A) -- (B) -- (C) -- cycle;  
\filldraw[blue!20, opacity=0.7] (Mid) circle (\radius);  
\end{scope}

\filldraw[red] (E) circle (0.3pt) node[anchor=south] {$x$};

\end{tikzpicture}
\end{figure}
\end{example}

In Step 2, we approximate $\BR$ by a function $\hat f$ whose graph is in $\Graph(\BR_\e)$ and whose fixed points are in $\cup_i V_i \backslash (\partial \S \cup \partial V_i)$. Now, we can proceed as in O'Neill to modify the function inside the $V_i$'s to obtain a function $f$ whose graph is still in $\Graph(\BR_\e)$ but whose fixed points are the $\s^{ij}$'s with index $r^{ij}$. 

Since we want to use $f$ to mimic a best-reply correspondence, the analysis is made easier if the function $f$ is locally affine around the $\s^{ij}$'s. This yields us the structure of a polymatrix game (one where the payoffs of a player are linear in the strategies of his opponents) at least locally around the equilibria of the game $\hat G^\e$, specified in the theorem.  The job of Step 1 is to ensure that it is possible to make $f$ locally affine. Local affineness, also, makes the task of constructing fixed points of $f$ with indices $+1$ or $-1$ simple, as one can define $f$ to be an affine homeomorphism locally with the corresponding $\s^{ij}$'s as fixed points.

We would like to view $f$ as a proxy for a best-reply correspondence of a perturbed game, as that would indicate how to perturb payoffs. There are two hurdles to doing so, which we deal with sequentially in our proof. The first problem is that the function $f$ maps into $\S \backslash \partial \S$ and therefore if $f_n(\s)$ is a ``best-reply'' to $\s_{-n}$, then every pure strategy of player $n$ is also a best reply to $\s_{-n}$.  The way around this problem is to say that $f(\s)$ is equivalent to a  strategy in an equivalent game $\tilde G$, thus avoiding making everything a best-reply.  Constructing a finite equivalent game, then, requires us to introduce finitely many points in each $\S_n$ as duplicate pure strategies and making $f(\s)$ mixtures over these points.  Geometrically speaking, this leads us to a triangulation $\T_n$ of $\S_n$ as carried out in Step 3. For each $\s$, player $n$'s ``best replies'' to $\s$ are among the vertices of the simplex that contains $f_n(\s)$ in its interior. As Example \ref{proxy} below demonstrates, resolving the first problem yields a new strategy space, call it $\tilde{\S}_n$, for each player $n$, in which each pure strategy is labeled at the vertex of the triangulation.

\begin{example}\label{proxy}Figure \ref{fproxy} illustrates the procedure we employ to use $f$ as a proxy for a best-reply correspondence in an equivalent game. The 2-simplex ABC (see (a)) represents the mixed strategy set of player $1$ in a finite game. Note that $f_1(\s) \neq \s_1$, so $\s = (\s_1, \s_2)$ is not a fixed point of $f$. We introduce a duplicate strategy $G$ in the mixed strategy set and from that produce a triangulation.  In this triangulation, the strategy $\s_1$ has ABG as a carrier, and the strategy $f_1(\s)$ has ACG as a carrier. In (d), we declare each vertex of the triangulation as a pure strategy in an equivalent strategy set, that is, a tetrahedron (recall the construction of the equivalent game illustrated in Figure \ref{equivalent}). Viewing $f_1(\s)$ and $\s_1$ as points in the tetrahedron, each of these strategies have distinct supports. If a perturbation of payoffs in the equivalent game can be made so that the support of $f_1(\s)$ can be made the set of pure best-replies to $\s_1$ (which is what is effectively meant by ``proxy''), then $\s_1$ cannot be an equilibrium in this perturbed equivalent game, since its support contains sub-optimal replies. 

\begin{figure}
\caption{Using $f$ to construct a best-reply proxy in an equivalent game}\label{fproxy}
\centering
\medskip
\begin{minipage}{0.3\textwidth}
    \centering
    \begin{tikzpicture}[scale=0.8] 
    \coordinate (A) at (0, 0);
    \coordinate (B) at (4, 0);
    \coordinate (C) at (2, 3.5);  
    \coordinate (S) at (2, 0.5);
    \coordinate (G) at (barycentric cs:A=1,B=1,C=1);  
    \coordinate (F1) at (barycentric cs:A=1,G=1,C=1); 

    \draw[thick] (A) -- (B) -- (C) -- cycle;

    \fill[black] (A) circle (1.5pt) node[anchor=north east] {A};
    \fill[black] (B) circle (1.5pt) node[anchor=north west] {B};
    \fill[black] (C) circle (1.5pt) node[anchor=south] {C};

    \fill[red] (F1) circle (1.5pt) node[anchor=south] {$f_1(\sigma)$};
    \fill[blue] (S) circle (1.5pt) node[anchor=south] {$\sigma_1$};
    \end{tikzpicture}
    \caption*{(a)}
\end{minipage}%
\hfill
\begin{minipage}{0.24\textwidth}
    \centering
    \begin{tikzpicture}[scale=0.8] 
    \coordinate (A) at (0, 0);
    \coordinate (B) at (4, 0);
    \coordinate (C) at (2, 3.5);  
    \coordinate (S) at (2, 0.5);
    \coordinate (G) at (barycentric cs:A=1,B=1,C=1);  
    \coordinate (F1) at (barycentric cs:A=1,G=1,C=1); 

    \draw[thick] (A) -- (B) -- (C) -- cycle;

    \fill[black] (A) circle (1.5pt) node[anchor=north east] {A};
    \fill[black] (B) circle (1.5pt) node[anchor=north west] {B};
    \fill[black] (C) circle (1.5pt) node[anchor=south] {C};
    \fill[black] (G) circle (1.5pt) node[anchor=south] {G};

    \fill[red] (F1) circle (1.5pt) node[anchor=south] {$f_1(\sigma)$};
    \fill[blue] (S) circle (1.5pt) node[anchor=south] {$\sigma_1$};
    \end{tikzpicture}
    \caption*{(b)}
\end{minipage}%
\hfill
\begin{minipage}{0.3\textwidth}
    \centering
    \begin{tikzpicture}[scale=0.8] 
    \coordinate (A) at (0, 0);
    \coordinate (B) at (4, 0);
    \coordinate (C) at (2, 3.5);  
    \coordinate (S) at (2, 0.5);
    \coordinate (G) at (barycentric cs:A=1,B=1,C=1);  
    \coordinate (F1) at (barycentric cs:A=1,G=1,C=1); 

    \filldraw[fill=red!40, opacity=0.3] (A) -- (C) -- (G) -- cycle; 
    \filldraw[fill=blue!20, opacity=0.3] (A) -- (B) -- (G) -- cycle; 

    \draw[thick] (A) -- (B) -- (C) -- cycle;

    \draw (G) -- (A);
    \draw (G) -- (B);
    \draw (G) -- (C);

    \fill[black] (A) circle (1.5pt) node[anchor=north east] {A};
    \fill[black] (B) circle (1.5pt) node[anchor=north west] {B};
    \fill[black] (C) circle (1.5pt) node[anchor=south] {C};
    \fill[black] (G) circle (1.5pt) node[anchor=south] {G};

    \fill[red] (F1) circle (1.5pt) node[anchor=south] {$f_1(\sigma)$};
    \fill[blue] (S) circle (2 pt) node[anchor=south] {$\sigma_1$};
    \end{tikzpicture}
    \caption*{(c)}
\end{minipage}%
\hfill
\begin{minipage}{0.24\textwidth}
    \centering
    \begin{tikzpicture}[scale=2] 
    \coordinate (A) at (0, 0);
    \coordinate (C) at (1.2, 0); 
    \coordinate (G) at (0.5, 1.5); 
    \coordinate (B) at (1.8, 0.5); 
    \coordinate (S) at (1, 0.8);

    \coordinate (F1) at ($(A)!1/3!(C)!1/3!(G)$); 

    \filldraw[fill=red!40, opacity=0.3] (A) -- (C) -- (G) -- cycle; 
    \filldraw[fill=blue!20, opacity=0.3] (A) -- (B) -- (G) -- cycle; 

    \draw[thick] (A) -- (B) -- (C) -- cycle;

    \draw[thick] (G) -- (A);
    \draw[thick] (G) -- (B);
    \draw[thick] (G) -- (C);

    \filldraw[black] (A) circle (0.5pt) node[anchor=north east] {A};
    \filldraw[black] (B) circle (0.5pt) node[anchor=west] {B};
    \filldraw[black] (C) circle (0.5pt) node[anchor=south] {C};
    \filldraw[black] (G) circle (0.5pt) node[anchor=west] {G};

    \filldraw[red] (F1) circle (0.5pt) node[anchor=north] {$f_1(\sigma)$};
    \filldraw[blue] (S) circle (0.5pt) node[anchor=north] {$\sigma_1$};
    \end{tikzpicture}
    \caption*{(d)}
\end{minipage}

\end{figure}

\end{example}

The second problem is that, unlike a best-reply correspondence, the coordinate function $f_n$ of player $n$ could depend on his strategy, $\s_n$.\footnote{One could say that the problem is more acute around the $\s^{ij}$'s where $f_n$ is exclusively a function of $\s_n$!} In any game, for any profile $\s$ of mixed strategies, the value of the best-reply correspondence of player $n$ does not depend on player $n$'s strategy $\s_{n}$. When generating the proxy for a best reply correspondence from the map $f$, recall that this map was obtained by approximating the best-reply correspondence of the original game by a continuous map and then modifying it through the use of the Hopf Extension Theorem. The resulting map $f$ unfortunately cannot be guaranteed to inherit that property of the best-reply correspondence. This impairs the possibility of converting  $f_n(\cdot)$ into a best-reply correspondence of an equivalent strategy set of player $n$ along the lines we illustrated in Example \ref{proxy}.

To deal with this problem, we consider an equivalent polytope-form  game $\tilde G$ in which each player $n$ chooses a strategy in $\tilde \S_n$ and one in $\tilde \S_{n+1}$ (where $n+1 = 1$ if $n = N$), and the strategy space of player $n$ is $\tilde \S_n \times \tilde \S_{n+1}$. The set $\tilde \S_n \times \tilde \S_{n+1}$ is a (non-simplicial) polytope. In order to define payoffs, we simply consider the first coordinates of each player and use the payoffs of the original game $G$: denote by $(\tilde \s_{n,0}, \tilde \s_{n,1}) \in \tilde \S_n \times \tilde \S_{n+1}$ a general strategy of player $n$ in $\tilde G$. The payoff from a profile $(\tilde \s_{n,0}, \tilde \s_{n,1})_{n \in \N}$ in $\tilde G$ to player $n$ is defined as $\tilde {G}_n((\tilde \s_{n,0},\tilde \s_{n,1})_{n \in \N}) = G_n((\tilde \s_{n,0})_{n \in \N})$. The construction is such that the choice by $n$ of a strategy in $\tilde \S_{n+1}$ is payoff-irrelevant for player $n$.  

Now the function $f$ can be ``lifted'' to this strategy space by making $f_n$ depend on $\s_{-n}$ and player $n-1$'s choice in $\S_n$ (where $n-1 = N$ if $n = 1$): since player $n$'s coordinate function $f_n$ might depend on $\S_n$, we will use player $n-1$'s copy of those coordinates.  We define a new map $\tilde f_n$ over the cartesian product of strategy sets of the polytope-form game (with the exception of player $n$'s) to $\tilde \S_n \times \tilde \S_{n+1}$ by $\tilde f_n ((\tilde \s_{m,0}, \tilde \s_{m,1})_{m \neq n})  = f_n(\tilde \s_{-n,0}, \tilde \s_{n-1,1}) \times \tilde \s_{n+1,0}$.\footnote{The precise definition of this map is slightly more nuanced than what is discussed here, but the main idea for its definition is to avoid the dependency on player $n$'s choices.}

Combining this idea along with the one illustrated in Example \ref{proxy}, we surpass the two main hurdles to obtaining a function that can be used as a proxy for a best-reply correspondence: steps 1 to 3 are preparatory steps to ensure the existence of a game equivalent to the original one, where a proxy for a best-reply correspondence can be constructed so as to satisfy the requirements of our main result (with precisely the $\s^{ij}$'s as fixed points and with the correct indices). Step 4 presents this equivalent game explicitly. 


The remaining steps of the proof show the precise sense in which $\tilde f$ is a proxy for a perturbed best-reply correspondence: they are responsible for constructing the payoff perturbation of the equivalent game whose best-reply correspondence will essentially match that proxy. 

This construction has two main phases: in the first one conducted in steps 5 and 6, for each $\tilde \s$, we define a payoff perturbation of the polytope-form game $\tilde G$ (constructed in step 4) such that for each $n$ the best replies to $\tilde \s$ are among the vertices of the carrier of $\tilde f_n(\tilde \s)$. The perturbation comes in the form of a bonus $\tilde g_{n}^0(\tilde \s)$ for the first coordinate of $n$'s strategy and $\tilde g_n^1(\tilde \s)$ for the second coordinate. For each $\tilde \s$, the finite, perturbed game is denoted by $\tilde{G} \oplus \tilde g^0(\tilde \s) \oplus \tilde g^1(\tilde \s)$. Step 5 details this construction. Step 6 analyses the relevant properties of the ``best-reply" correspondence that assigns to each $\tilde \s$ the set of best replies to it in the perturbed game  $\tilde{G} \oplus \tilde g^0(\tilde \s) \oplus \tilde g^1(\tilde \s)$: note that the payoff perturbations in Step 5 depend on the strategy profile $\tilde{\s}$ and therefore may be different depending on which profile is played. Hence, though the correspondence that assigns to each $\tilde \s$ the set of best replies to it in $\tilde{G} \oplus \tilde g^0(\tilde \s) \oplus \tilde g^1(\tilde \s)$ matches our proxy $\tilde f$, the correspondence is not (yet) the best-reply correspondence of a \textit{finite} game.  

The final phase of the construction uses the maps $\tilde g^0$ and $\tilde g^1$ to finally construct a finite game with the required properties. This phase is the most technical part of the proof and providing a thorough intuition for it cannot be properly done without introducing an inordinate amount of notations and definitions. In spite of that, we would like to provide some insight into the final part of the argument. 

The perturbation  $\tilde g^0(\cdot)$ maps $\tilde g^1(\cdot)$ are simply continuous maps defined over the strategy space of the polytope-form game $\tilde G$, namely, $\prod_{n \in \N} (\tilde{\S}_n \times \tilde{\S}_{n+1})$. If maps $\tilde g^0(\cdot)$ and $\tilde g^1(\cdot)$ were in addition multilinear maps of the strategy profiles of this game, then they would immediately define payoff perturbations allowing for the definition of a perturbed polytope-form game. Step 7 provides the preliminary work for generating approximations of the maps $\tilde g^0(\cdot)$ and $\tilde g^1(\cdot)$ so as to be able to define the desired polytope form. By triangulating the strategy sets of each player with sufficiently small diameter, we approximate the maps $\tilde g^0(\cdot)$ and $\tilde g^1(\cdot)$ by maps which are locally multilinear in each product of simplices of the triangulation. Each newly introduced vertex of the triangulation is declared to be a new strategy of an equivalent game $\hat G$ to $\tilde G$---in the same fashion as exemplified in Figure \ref{fproxy}---and this game can now be perturbed so as to achieve the required properties.

Step 8 defines the perturbation of the game $\hat G$, which has five components.  The first two ($\hat G^0$ and $\hat G^1$) correspond to the perturbations that are induced by the multiaffine versions of the maps $\tilde g^0$ and $\tilde g^1$.  Two other perturbations ($\hat g^0$ and $\hat g^1$) are required to ensure that optimal strategies only mix over points that are vertices of a cell of the triangulation (this is illustrated in more detail in Example \ref{final}). The fifth and final component of payoff perturbations is required to ensure uniqueness of the equilibria (projecting to $\sigma^{ij}$'s). As explained, passing to equivalent games involves introducing duplicate strategies, which may result in sets of mixed strategies that are equivalent to a point $\s^{ij}$ in the original game. One therefore needs to select one among the many profiles of strategies included in the equivalent set to be the isolated equilibrium of the equivalent game - and the payoff perturbations are constructed to achieve this goal. 
Step 9 completes the proof of the main theorem by showing that this perturbed game is the finite normal-form game we were after.

Since the proof involves a lot of notation, some of which are step-specific and others not, we end this section with a table of the notations that are used across multiple steps. The table should be used as a consultation tool throughout the reading of the proof, in case some notation that has already been encountered is forgotten.   \begin{center}
	\renewcommand{\arraystretch}{1.5}
	\captionof{table}{Table of Notations}
	\begin{tabular}{ | m{5cm} | m{11cm} | } 
		\hline
		\textbf{Notation} & \textbf{Reference and/or Definition} \\
		\hline
		$Y^{ij}_n$ & Simplex contained in $\S_n$, with $\s^{ij}_n$ as barycenter (see Step 1)  \\ 
		\hline
		$X^{ij}_n$ & Simplex contained in $Y^{ij}_n$, with $\s^{ij}_n$ as barycenter (see Step 1) \\ 
		\hline
		$f^{ij}_n$ & $f^{ij}_n: X^{ij}_n \to Y^{ij}_n$; see Lemma \ref{lem f ij}  \\ 
		\hline
		$\BR$ & Best-reply correspondence of game $G$ \\ 
		\hline
		$\BR_{\e}$ & $\e$-Best-reply correspondence of game $G$ (see Step 1) \\ 
		\hline
		$f$ & $f: \S \to \S$ (See Lemma \ref{lem function f}, Step 2) \\
		\hline
		$\T_n$ & Triangulation of $\S_n$ (see Lemma \ref{lem triangulation T}, Step 3) \\
		\hline 
		$T^{ij}_n$ & Simplex of $\T_n$ contained in $X^{ij}_n$, containing $\s^{ij}_n$ as its barycenter  (see Lemma \ref{lem triangulation T}, Step 3) \\
		\hline
		$\tilde S_{n,0}$ & Vertices of the triangulation $\T_n$ (see Step 4) \\ 
		\hline
		$\tilde S_{n,1}$ & Vertices of the triangulation $\T_{n+1}$ (see Step 4) \\
		\hline
		$\tilde \S_{n,k}$ & $\tilde \S_{n,k} \equiv \D(\tilde S_{n,k}), k =0,1$ (See Step 4) \\
		\hline
		$\tilde G$ & Polytope-form game where the strategy set of each player $n$ is $\tilde \S_n \equiv \tilde \S_{n,0} \times \tilde \S_{n,1}$ (see Step 4) \\
		\hline
		$\tilde BR_\e$ & $\e$-Best-reply correspondence of game $\tilde G$ (see Step 4) \\
		\hline
		$\tilde \b_{n,k}$ & $\tilde \b_{n,k}: \S_n \to \tilde \S_{n,k}$ is the barycentric map: it associates to a point $\s_n$ its barycentric coordinates in the triangulation $\T_n$, viewed as a face of $\tilde \S_{n,k}, k=0,1$ (see Step 4) \\
		\hline
		$\phi_{n,k}$ &  $\phi_{n,k}: \tilde \S_{n,k} \to \S_{n+k}$ projects the point $\tilde \s_{n,k}$ to its location in $\S_{n+k}, k=0,1$ (see Step 4) \\
		\hline
		$\tilde \S^*_{n,k}$ & $\tilde \S^*_{n,k} \equiv \tilde \b_{n,k}(\S_{n+k}), k=0,1$ (see step 4) \\
		\hline
		$\mathcal{S}^*_{n,k}$ & Subcomplex of the face complex of $\D(\tilde S_{n,k})$ whose space is $\tilde \S^*_{n,k}, k=0,1$ (see Step 4) \\
		\hline
		$\tilde T^{ij}_{n,k}$ & $ \tilde T^{ij}_{n,k} \equiv \tilde \b_{n,k}((f^{ij}_{n+k})^{-1}(T^{ij}_{n+k})), k=0,1$ (see Step 4) \\
		\hline
	
		\end{tabular}
\end{center}
\begin{center}
\renewcommand{\arraystretch}{1.5}
\begin{tabular}{ | m{5cm} | m{11cm} | } 

\hline
\textbf{Notation} & \textbf{Reference and/or Definition} \\
\hline
$\tilde f_{n,0}$ & $$\tilde f_{n,0}: \tilde \S_{-n} \to \tilde \S_{n,0},$$  $$\tilde f_{n,0}(\tilde \s_{-n}) \equiv \tilde \beta_{n,0}(f_n(\phi_{-n, 0}(\tilde \s_{-n,0}), \phi_{n-1, 1}(\tilde \s_{n-1})))$$  (see Step 4 and Lemma \ref{lem function f tilde}) \\
		\hline
		$\tilde f_{n,1}$ & $\tilde f_{n,1}(\tilde \s_{n+1}) \equiv \tilde \b_{n,1}(\phi_{n+1,0}(\tilde \s_{n+1}))$ (see Step 4 and Lemma \ref{lem function f tilde}) \\
		\hline
		$\tilde g^0(\cdot)$ & $\tilde g^0: \tilde \S \to \prod_n \Re^{\tilde S_{n,0}}$ (see Step 5) \\
		\hline
		$\tilde g^1(\cdot)$ & $\tilde g^1: \tilde \S \to \prod_n \Re^{\tilde S_{n,1}}$ (see Step 5) \\
		\hline
		$r_n(\cdot)$ & $r_n: \S_{-n} \to \Re$ returns the best payoff player $n$ obtains against $\s_{-n}$ (see Step 5) \\
		\hline
		$\tilde \T_n$ & Triangulation of $\D(\tilde S_{n,0}) = \D(\tilde S_{n-1,1})$ (see Step 7) \\
		\hline
		$\hat S_{n,0}$ & $\hat S_{n,0}$ the set of vertices of $\tilde \T_n$, $\hat S_{n,0} = \hat S_{n-1,1}$ (see Step 8) \\
		\hline
		$\hat T^{ij}_n$ & Carrier of $\tilde \s^{ij}_n$ in $\tilde \T_n$  (see Step 7 and Lemma \ref{lem triangulation T tilde}) \\
		\hline
		$\hat \phi_{n,k}$ & $\hat \phi_{n,k}: \hat \S_{n,k} \to \tilde \S_{n,k}$ projects strategies from $\hat \S_{n,k}$ to $\tilde\S_{n,k}$ (see Step 8) \\ 
		\hline
		$\hat \mu_{n,k}$ & $\hat \mu_{n,k}: \hat \S_n \to \hat \S_{n,k}$ computes the marginal of $\hat \s_n$ over $\S_{n,k}$ (see Step 8) \\
		\hline 
		$\hat G$ & Normal-form game with pure strategy set of player $n$ equals to $\hat S_n \equiv \hat S_{n,0} \times \hat S_{n,1}$ and payoff $\hat G_n(\hat \s) = \hat G_n(\hat \phi \circ \hat \mu (\hat \s))$ (see Step 8)\\
		\hline
		
	\end{tabular}
\end{center}

\medskip

\section{Proof of Theorem \ref{main thm}}\label{sec proof}

As in the statement of the theorem, we have the following notation, used uniformly throughout the section. For each component  $C_i$, $i = 1, \ldots, k$, we have a finite (possibly empty) set $J_i$, and a finite set of integers  $r^{ij} \in \{ \, +1, -1\, \}$,  $j \in J_i$, such that $\sum_{j \in J_i} r^{ij} = c_i$, with the sum being zero if $J_i$ is an empty set.  

\medskip

\noindent {\bf Step 1.} Fix a finite number of completely-mixed strategy profiles $\s^{ij}$, $i = 1, \ldots, k$ and $j \in J_i$ such that for each $n$ the $\s_n^{ij}$'s are distinct mixed strategies.  Locally around each $\s^{ij}$, we construct an affine map that has $\s^{ij}$ as its unique fixed point and assigns it the index $r^{ij}$.  In Step 2, these maps are then extended to the whole of $\S$ (not affinely though) without introducing additional fixed points. 

For each $n$, and $i,j$, let $X_n^{ij}$ and $Y_n^{ij}$ be two full-dimensional simplices in $\S_n$ with $\s_n^{ij}$ as their barycenters and $X_n^{ij}\subseteq \text{int}(Y_n^{ij})$.  Let $X^{ij} = \prod_n X_n^{ij}$ and $Y^{ij} = \prod_n Y_n^{ij}$. We will assume that the simplices $Y_n^{ij}$ are small enough that the $Y^{ij}$'s are pairwise disjoint. For now, these are the only resctrictions we choose for $Y_n^{ij}$ and $X_n^{ij}$ in order to obtain Lemma \ref{lem f ij}. In Step 2, we will specify their choice further.
  
\begin{lemma}\label{lem f ij}
	For each $i,j$ and $n$, there is an affine homeomorphism $f_n^{ij}: X_n^{ij} \to Y_n^{ij}$ such that:
	\begin{enumerate}
		\item $\s_n^{ij}$ is the unique fixed point of $f_n^{ij}$;
		\item letting $f^{ij} = \prod_n f_n^{ij}: X^{ij} \to Y^{ij}$, the index of $\s^{ij}$ under $f^{ij}$ is $r^{ij}$;
		\item for all $\s_n \neq \s^{ij}_n$, $\Vert f^{ij}_n (\s_n) - \s^{ij}_n \Vert > 0$.
	\end{enumerate}  	
\end{lemma}

\begin{proof}
	The result for the case $N =1$ was proved in \cite{GP2022}.  Use their result for each $n$, assigning index $r^{ij}$ to player 1's map $f_1^{ij}$ and $+1$ for the maps of all other players. The index of  $\s^{ij}$ under $f^{ij}$ is $r^{ij}$ by the multiplication property of the index. 
\end{proof}


\medskip

\noindent {\bf Step 2.} This step gives us a version of O'Neill's theorem with three main differences. (1) It applies to the $\BR$ correspondence. (2) The points chosen are in the interior of $\S$ and close to, but not necessarily in, $C_i$. (3) We insist that the map be locally affine around the fixed points. 

\medskip

Let $\Graph(\BR)$ be the graph of the best-reply correspondence $\BR: \S \twoheadrightarrow \S$ of the game $G$.  For each $\e > 0$, let $\BR_\e$ be the $\e$-best-reply correspondence, i.e., for each $\s$, $\BR_\e(\s)$ is the set of strategy profiles $\t$ such that for each $n$, $G_n(\s_{-n}, \t_n) > \max_{s_n} G_n(\s_{-n}, s_n) - \e$. The graph of $\BR_\e$, denoted $\Graph(\BR_\e)$, is a neighborhood of $\Graph(\BR)$. (In fact, the sets $\Graph(\BR_\e)$ form a basis of neighborhoods of $\Graph(\BR)$.) From now on, we fix $\e > 0$ as required by the main theorem.

Choose $\d > 0$ such that the $\d$-neighborhoods $V_i$ of $C_i$ satisfy the following properties: $V_i \cap V_j = \emptyset$ for all $i \neq j$; and for all $i$ and $\s \in V_i$,  $(\s, \s)$ belongs to  $\Graph(\BR_{\e})$. These $V_i$'s are the ones specified in the statement of Theorem \ref{main thm}.  These, too, are fixed from now on.  Pick points $\s^{ij}$, $i = 1, \ldots, k$, $j \in J_i$ as required by the same theorem and fix them from now on as well.  If necessary by considering an equivalent game $\bar G$ obtained by adding duplicate strategies, we can assume that the $\s_n^{ij}$'s are all distinct for each $n$ and with disjoint supports.  Indeed, consider an equivalent game $\bar G$ where we duplicate every pure strategy of every player. Let $\bar \S$ be the strategy space in $\bar G$ and $\phi: \bar \S \to \S$ be the map that sends strategies in $\bar G$ to equivalent profiles in $G$.  Let $\bar C_i = \phi^{-1}(C_i)$ be the equilibrium component of $\bar G$ for each $i$, and let $\bar V_i \equiv \phi^{-1}(V_i)$ be the corresponding neighborhood of $\bar C_i$.  For each $\s^{ij}$, we can pick a point $\bar \s^{ij} \in \phi^{-1}(\s^{ij})$ such that for each $n$, the $\bar \s_n^{ij}$'s are distinct mixed strategies with disjoint supports. Because $\bar G$ is equivalent to $G$, each $\bar C_i$ has the same index as $C_i$ (cf. Theorem 5 in \cite{GW2005}). Therefore, in what follows, we can replace $G$ with $\bar G$.  To simplify notation, we will refer to the game still as $G$ and just assume that the chosen $\s^{ij}$'s involve distinct mixed strategies for each player. We ilustrate this procedure in Example \ref{distinct}.

\begin{example}\label{distinct}

In Figure \ref{duplication}, (a) represents a 1-simplex with vertices $A$ and $B$ where we singled-out point $P$, representing a mixed strategy. By duplicating the vertex $A$ to $A'$ and $B$ to $B'$, we represent in (b) a tetrahedron, where the points $P_1$ and $P_2$ are singled-out duplicates of point $P$, with disjoint supports.

\begin{figure}[ht]
\caption{Duplicating pure strategies and choosing points with distinct supports}\label{duplication}
\centering

\begin{minipage}[t]{0.45\textwidth}
    \centering
    \begin{tikzpicture}[scale=2]
        \coordinate (A) at (0, 0);
        \coordinate (B) at (2, 0);

        \draw[thick] (A) -- (B);

        \coordinate (M) at ($(A)!0.5!(B)$); 

        \filldraw[black] (A) circle (1pt) node[anchor=north] {A};
        \filldraw[black] (B) circle (1pt) node[anchor=north] {B};

        \filldraw[red] (M) circle (1.5pt) node[anchor=south] {P};
    \end{tikzpicture}
    \caption*{(a)}
\end{minipage}%
\hfill
\begin{minipage}[t]{0.45\textwidth}
    \centering
    \begin{tikzpicture}[scale=2.5]

        \coordinate (A) at (0, 0);
        \coordinate (B) at (1.25, 0);   
        \coordinate (C) at (0.75, 1.5);    
        \coordinate (G) at (1.5, 0.25);      

        \draw[thick] (A) -- (B) -- (C) -- cycle;

        \draw[thick] (A) -- (G);
        \draw[thick] (B) -- (G);
        \draw[thick] (C) -- (G);

        \filldraw[black] (A) circle (1pt) node[anchor=north] {A};
        \filldraw[black] (B) circle (1pt) node[anchor=south] {B};
        \filldraw[black] (C) circle (1pt) node[anchor=south] {$A'$}; 
        \filldraw[black] (G) circle (1pt) node[anchor=south] {$B'$}; 

        \coordinate (M_AB) at ($(A)!0.5!(B)$);
        \coordinate (M_CG) at ($(C)!0.5!(G)$);

        \filldraw[red] (M_AB) circle (1pt) node[anchor=north] {$P_2$};
        \filldraw[red] (M_CG) circle (1pt) node[anchor=south] {$P_1$};

    \end{tikzpicture}
    \caption*{(b)}
\end{minipage}

\end{figure}

\end{example}

\begin{lemma}\label{lem function f}
	There exists a continuous function $f: \S \to \S$ with the following properties:
	\begin{enumerate}
		\item the graph of $f$ is contained in $\Graph(\BR_\e)$;
		\item the fixed points of $f$ are the points $\s^{ij}$ and the index of each $\s^{ij}$ is $r^{ij}$; 
		\item for each $n$ and $i,j$, there exist a polyhedral subdivision $\mathcal{X}_n$ of $\S_n$ and full-dimensional simplices $X_n^{ij}$ and $Y_n^{ij}$ such that:
		\begin{enumerate}
			\item $\s_n^{ij}$ is the barycenter of both $X_n^{ij}$ and $Y_n^{ij}$;
			\item $X_n^{ij} \subset \text{int}(Y_n^{ij}) \subset \S_n \backslash \partial \S_n$;
			\item $\prod_n Y_n^{ij}\times \prod_n Y_n^{ij} \subset \Graph(\BR_\e)$;
			\item the restriction $f^{ij}$ of $f$ to $\prod_n X_n^{ij}$ is a product map $\prod_n f_n^{ij}$ where $f_n^{ij}$ is an affine homeomorphism of $X_n^{ij}$ onto $Y_n^{ij}$;
			\item  $X^{ij}_n$ is the space of a subcomplex of $\mathcal{X}_n$ and the carrier of $\s^{ij}_n$ in $\mathcal{X}_n$ has dimension dim($\S_n$).
		\end{enumerate}
	\end{enumerate}    
\end{lemma}

\begin{remark} As it will be clear from the proof of the lemma and illustrated in Example \ref{example lem f}, the polyhedral subdivision $\mathcal{X}_n$ is constructed by ``extending'' the faces of each $X^{ij}_n$, thus producing a polyhedral subdivision of $\S_n$ (see Figure \ref{extendingsub}). Point $(e)$ of the Lemma implies that $X^{ij}_n$ is a union of cells of the polyhedral subdivision. The simplices $Y^{ij}_n$ play no role in the construction of $\mathcal{X}_n$ and are neither a cell nor the space of a subcomplex of $\mathcal{X}_n$ necessarily. \end{remark} 

\begin{proof}
 	By Corollary 1 and Axiom 2 in \cite{AM1989}, there  exists a function $\hat f: \S \to \S$ such that

	\begin{enumerate}
		\item the graph of $\hat f$ is contained in $\Graph(\BR_{\e})$;
		
		\item all its fixed points are contained in $\cup_i (V_i \backslash \partial_{\S} V_i)$;
	
		\item the index of $\hat f$ over $V_i$ is $c_i$ for each $i$.
		\end{enumerate}	

  If necessary by replacing $\hat f$ with a suitable convex combination $(1-\a) \hat f + \a \s^{\circ}$, where $\sigma^{\circ}$ is a completely mixed profile in $\S$ and $\a>0$ is sufficiently small, we can assume furthermore that the fixed points of $\hat f$ are contained in $\S \backslash \partial \S$.

    Let $C_i(\hat f)$ be the set of fixed points of $\hat f$ in $V_i$. As $V_i$ is the $\d$-neighborhood of the component $C_i$, it is semialgebraic (cf. Proposition 2.2.8 in \cite{BCR1998}) and $V_i \backslash (\partial_{\S_i} V_i \cup \partial \S)$ is connected. Therefore for all $\eta > 0$ sufficiently small, the set of points in $V_i$ whose distance from $\partial_{\S_i} V_i \cup \partial \S$ is strictly greater than $\eta$ is a connected set.  As the points $\s^{ij}$ and the set $C_i(\hat f)$ lie in $V_i \backslash (\partial_{\S_i} V_i \cup \partial \S)$,  we can choose $\eta > 0$ small enough such that for each $i$, letting $U_i$ be the set of $\s \in V_i$ such that $d(\s, \partial_{\S_i} V_i \cup \partial \S) > 2\eta$, we have that $U_i \backslash \partial U_i$ is connected and contains $C_i(\hat f)$ as well as the points $\s^{ij}$. Furthermore, by choosing an even smaller $\eta$, if necessary, we can assume that $(\s, \t) \in \Graph(\BR_\e)$ for $\s \in V_i$ and $\t$ within $2\eta$ of $\s$.

	Use Urysohn's lemma to construct a function  $\g: \S \to [\eta, 1]$ that equals 1 outside the interiors of $V_i$'s and $\eta$ on $U_i$. Replace the function $\hat f$ with the function $\hat f^1  \equiv \g \hat f + (1-\g)\text{Id}$. The graph of $\hat f^1$ is also contained in $\Graph(\BR_{\e})$ and its fixed points are still the points in the $C_i(\hat f)$'s. Also, $\Vert \hat f^1(\s) - \s \Vert \le 2\eta$ if $\s \in U_i$ for some $i$. 
		
 	For each $i,j$ and $n$ choose simplices $X_n^{ij}$  and $Y_n^{ij}$ such that:

\begin{enumerate}
  	\item[(4)] 	$X_n^{ij}$ is contained in the interior of $Y_n^{ij}$, and $Y^{ij}$ is contained in the interior of $U_i$; 
	\item[(5)]  $\s_n^{ij}$ is the barycenter of both $X_n^{ij}$ and $Y_n^{ij}$;
	\item[(6)]  $Y_n^{ij}$ is contained in the $\eta$-radius ball around $\s^{ij}$;  
	\item[(7)] $\{Y_n^{ij}\}_{i,j}$ are pairwise disjoint for each $n$.    
\end{enumerate}

Each dimension dim($\S_n) -1$ face $F^{ij}_n$ of $X^{ij}_n$ is contained in a unique hyperplane $H_{F^{ij}_n}$ in $\Re^{S_n}$ orthogonal to $\S_n$. Denote by $H^{*}_{F^{ij}_n}$, $* \in \{+,-\}$, the positive and negative half-spaces defined by this hyperplane. For a fixed $X^{ij}_n$, if a hyperplane $H_{F^{ij}_n}$ intersects $\s^{\hat i \hat j}_n$ for $\hat i \neq i$ or $\hat j \neq j$, then the vertices of $F^{ij}_n$ are not in general position in $\S_n$, i.e., they belong to a subspace of $\S_n$ of dimension strictly lower than dim$(\S_n)$. Therefore, choosing the vertices of $X_n^{ij}$ conveniently in the interior of $Y^{ij}_n$, we can assume that $H_{F^{ij}_n}$ does not intersect $\s^{\hat i \hat j}_n$ for $\hat i \neq i$ or $\hat j \neq j$. We assume this holds for each $i,j,n$. We can now define $\mathcal{X}_n$. Consider the intersection of the hyperplane $H_{F^{ij}_n}$ with $\S_n$, for each $i,j$. This defines a polyhedral subdivision where a cell of the subdivision is given by the intersection of finitely many half-spaces $H^{*}_{F^{ij}_n}$.\footnote{In the Appendix, we elaborate on polyhedral subdivisions constructed through this procedure.} Due to our assumption, the carrier of $\s^{ij}_n$ in this subdivision is a cell of dimension dim$(\S_n)$. By construction of this subdivision, $X^{ij}_n$ is the space of a subcomplex of $\mathcal{X}_n$. This proves 3(e).

Use Step 1 to construct for each $n$ and $i,j$ an affine homeomorphism $f_n^{ij}: X_n^{ij} \to Y_n^{ij}$ satisfying property (1) of Lemma \ref{lem function f} and such that $f^{ij} = \prod_n f_n^{ij}$ satisfies property (2) of Lemma \ref{lem function f}.  Letting $A$ be the affine space generated by $\S - \S$,  define $d: (\S \backslash \cup_i  (U_i \backslash \partial U_i)) \cup (\cup_{i,j} X^{ij}) \to A$ by $d(\s) = \s - \hat f^1(\s)$ if $\s \notin \cup_i (U_i \backslash \partial U_i)$; $d(\s) = \s - f^{ij}(\s)$ if $\s \in X^{ij}$.  The only zeros of $d$ are the $\s^{ij}$'s.   The set $U_i$ is semi-algebraic and $U_i \backslash \partial U_i$ is an open and connected set.  Therefore, $(U_i, \partial U_i)$ is a pseudomanifold with boundary. Hence, $U_i \backslash \cup_j(X^{ij} \backslash \partial X^{ij})$ is a pseudo-manifold with boundary $(\cup_j \partial X^{ij}) \cup \partial U_i$ and the restriction of $d$ to this boundary has degree zero. As there are at least two players with at least two pure strategies each, the dimension of $\S$ is at least two. Therefore, by the Hopf Extension Theorem, $d$ extends to the sets $U_i$ in such a way that its norm over $\cup_i U_i$ is still no more than $2\eta$ and it has no additional zeros.\footnote{The version of the Hopf Extension Theorem we use here is Corollary 18, Chapter 8 in \cite{ES1966}. We note this Corollary is stated for singular cohomology (with coefficients in $\mathbb{Z}$), and provides a condition for the extension of $d: (\cup_j \partial X^{ij}) \cup \partial U_i \to A -\{0\}$ to  $U_i \backslash \cup_j(X^{ij} \backslash \partial X^{ij})$.  The condition is that $\d(d^*(s^*)) =0$, where $s^*$ is a generator of $H^m(A - \{0\})$ ($A - \{0\}$ being a homotopy sphere), where $m = \text{dim($A$)} -1$ and $\d$ is the coboundary cohomology morphism of the long exact sequence of the pair $(U_i \backslash \cup_j(X^{ij} \backslash \partial X^{ij}),(\cup_j \partial X^{ij}) \cup \partial U_i)$. If the degree of $d: (\cup_j \partial X^{ij}) \cup \partial U_i \to A -\{0\}$ is zero (because it is the displacement of $\hat f^1$), then by definition of the singular cohomology functor, $d^*(s^*) =0$, which implies $\d(d^*(s^*)) =0$.}  Thus, we have a displacement map defined over the whole of $\S$, denoted still by $d$. Define $f$ as $\text{Id} - d$.  
 	
 	Outside $\cup_i U_i$,  $d$ is the displacement of $\hat f^1$, whose graph is contained in $\Graph(\BR_\e)$; on the $U_i$'s $d$ has norm $2\eta$ or less, which by the choice of $\eta$, therefore, means the graph of $f$ is contained in $\Graph(\BR_\e)$, proving point (1) of the Lemma.  Point 3(c) of the lemma follows from the fact that the diameter of $Y^{ij}$ is at most $2\eta$; all other points of the lemma follow from the construction of the function $f$. 
\end{proof}

\begin{example} We explain and illustrate our claim on the degree of $d$ over $(\cup_{j} \partial X^{ij}) \cup \partial U_i$ being zero. One can apply Theorem 5.1 in \cite{O1953} to obtain this result: the degree of $d$ over $\partial U_i$ is determined by the index of $\hat{f}^1$ in $U_i$ and is $c_i$.  The degree of $d$ over $\partial X^{ij}$ is determined by the index of $f^{ij}$ over $X^{ij}$ which is $r^{ij}$ by construction. This degree is measured with respect to the orientation induced on $\partial X^{ij}$ by $X^{ij}$. But the orientation induced by $U_i \backslash \cup_j(X^{ij} \backslash \partial X^{ij})$  over $\partial X^{ij}$ is opposite to the one with respect to which $f^{ij}$ is measured, by the definition of orientability. According to this orientation, the index of $f^{ij}$ is opposite, hence, $-r^{ij}$. Therefore, since $\sum_j r^{ij} = c_i$,  the restriction of $d$ to $(\cup_j \partial X^{ij}) \cup \partial U_i$ has degree zero. 

In (a) Figure \ref{orient}, we illustrate a counterclockwise orientation induced on the boundary of the disc, and the counterclockwise orientations on each of the boundaries of the smaller red discs centered in points $P_1$ and $P_2$. These orientations agree, since the ones on the boundary of the smaller discs are induced (by excision) from the one on the boundary of the blue disk. But when orienting the boundary of the blue disc with the two smaller red discs removed, note that the inner boundaries are now clockwise oriented, by the definition of orientability. So if a map has degree $+2$ measured according to the orientation on the boundary of the blue disc (or according to the sum of the degrees measured according to the orientation on the boundary of the red discs, say, each being $+1$), then, when the red discs are removed, the orientation on the inner boundary is now clockwise (therefore implying the degree of the map, with respect to this orientation, is $-1$). So the degree on the outer-boundary is $+2$, and on the inner-boundaries is $-1$ each, resulting in the degree of zero in the union of all boundaries.

\begin{figure}[ht]
\caption{Orientation}\label{orient}
\medskip
\medskip
\begin{minipage}{0.4 \textwidth}
\centering
\begin{tikzpicture}[scale=0.8]

\filldraw[fill=blue!20, draw=red, thick] (0,0) circle (3);

\filldraw[red] (1, -1) circle (2pt);
\filldraw[red] (-1.5, 0.5) circle (2pt);

\draw[->, red, thick, >=stealth, line width=2pt] (3, 0) arc [start angle=0, end angle=270, radius=3];

\node[above right] at (1, -1) {$P_1$};
\node[above left] at (-1.5, 0.5) {$P_2$};

\filldraw[fill=red!20, draw=red, thick] (1,-1) circle (0.8);
\draw[->, red, thick, >=stealth, line width=1.5pt] (1.8, -1) arc[start angle=0, end angle=270, radius=0.8];

\filldraw[fill=red!20, draw=red, thick] (-1.5,0.5) circle (0.7);
\draw[->, red, thick, >=stealth, line width=1.5pt] (-0.8, 0.5) arc[start angle=0, end angle=270, radius=0.7];

\filldraw[red] (1, -1) circle (2pt);
\filldraw[red] (-1.5, 0.5) circle (2pt);

\node[above right] at (1, -1) {$P_1$};
\node[above left] at (-1.5, 0.5) {$P_2$};

\end{tikzpicture}
\caption*{(a)}
\end{minipage} 
\begin{minipage}{0.3 \textwidth}
	\centering
	\begin{tikzpicture}[scale=0.8]

\filldraw[fill=blue!20, draw=red, thick] (0,0) circle (3);

\filldraw[red] (1, -1) circle (2pt);
\filldraw[red] (-1.5, 0.5) circle (2pt);

\draw[->, red, thick, >=stealth, line width=2pt] (3, 0) arc [start angle=0, end angle=270, radius=3];

\node[above right] at (1, -1) {$P_1$};
\node[above left] at (-1.5, 0.5) {$P_2$};

\filldraw[fill=white!0, draw=red, thick] (1,-1) circle (0.8);
\draw[<-, red, thick, >=stealth, line width=1.5pt] (1.8, -1) arc[start angle=0, end angle=270, radius=0.8];

\filldraw[fill=white!0, draw=red, thick] (-1.5,0.5) circle (0.7);
\draw[<-, red, thick, >=stealth, line width=1.5pt] (-0.8, 0.5) arc[start angle=0, end angle=270, radius=0.7];

\filldraw[red] (1, -1) circle (2pt);
\filldraw[red] (-1.5, 0.5) circle (2pt);

\node[above right] at (1, -1) {$P_1$};
\node[above left] at (-1.5, 0.5) {$P_2$};

	\end{tikzpicture}
	\caption*{(b)}
\end{minipage} 
\end{figure}

\end{example}

\begin{example}\label{example lem f} Figure \ref{extendingsub} illustrates how the polyhedral subdivision of Lemma \ref{lem function f} is obtained. The two points in red in Figure (a) can be viewed as two of the $\s^{ij}_n$'s, first picked in the interior of the simplex. We omit any reference to $Y^{ij}_n$'s in this construction because they do not play a role in obtaining the subdivision.  In (b) we choose the vertices for each of the $X^{ij}_n$'s that should have the fixed red points as their barycenters. In (c), the 2-simplices $A1-B1-C1$ and $A2-B2-C2$ should be viewed as two distinct $X^{ij}_n$'s  of player $n$.  When the vertices $A1,B1,C1$ and $A2,B2,C2$ are well-chosen, for previously fixed points $BC1$ and $BC2$, then they define 2-simplices such that, when their 1-dimensional faces are extended as in the figure, they do not intersect any of the points $BC1$ and $BC2$. The result of the process is a polyhedral subdivision as depicted in (d), where $BC1$ and $BC2$ have 2-dimensional carriers and where simplices $A1-B1-C1$ and $A2-B2-C2$ are unions of cells of the subdivision, i.e., both are spaces of a subcomplex of the the polyhedral subdivision: in particular, $A1-B1-C1$ is a cell of the resulting polyhedral subdivision, whereas $A2-B2-C2$ is not. 


\begin{figure}[hb]
\caption{An illustration for the polyhedral subdivision of Lemma \ref{lem function f}}\label{extendingsub}
\medskip
\begin{minipage}[t]{0.24\textwidth}
    \centering
\begin{tikzpicture}[scale=5]

\coordinate (A) at (0, 0);
\coordinate (B) at (1, 0);
\coordinate (C) at (0.5, 0.866);  
\coordinate (AC) at (0.3186,0.5519);
\coordinate (AB) at (0.65,0);
\coordinate (AC1) at (0.1444,0.25); 
\coordinate (CB) at (0.8556,0.25);
\coordinate (AB1) at (0.05,0);
\coordinate (BC) at (0.5341,0.8069);

\coordinate (CB1) at (0.7254,0.4757);
\coordinate (AB2) at (0.44,0);

\coordinate (AB3) at (0.0577,0.1);
\coordinate (BC3) at (0.9422,0.1);

\coordinate (A1) at (0.2, 0.25);
\coordinate (B1) at (0.5, 0.25);
\coordinate (C1) at (0.35, 0.5);

\coordinate (A2) at (0.5, 0.1);
\coordinate (B2) at (0.8, 0.1);
\coordinate (C2) at (0.65, 0.35);

\coordinate (CA3) at (0.4217,0.7304); 
\coordinate (BA4) at (0.86,0);

\draw[thick] (A) -- (B) -- (C) -- cycle;



\coordinate (BC1) at (barycentric cs:A1=1,B1=1,C1=1);  
\coordinate (BC2) at (barycentric cs:A2=1,B2=1,C2=1);  

\filldraw[red] (BC1) circle (0.25pt) node[anchor=south, font=\scriptsize] {BC1};
\filldraw[red] (BC2) circle (0.25pt) node[anchor=north, font=\scriptsize] {BC2};

\filldraw[black] (A) circle (0.2pt) node[anchor=north east, font=\scriptsize] {A};
\filldraw[black] (B) circle (0.2pt) node[anchor=north west, font=\scriptsize] {B};
\filldraw[black] (C) circle (0.2pt) node[anchor=south, font=\scriptsize] {C};


\end{tikzpicture}
\caption*{(a)}
\end{minipage}
\hfill
\begin{minipage}[t]{0.24 \textwidth}
	\centering
\begin{tikzpicture}[scale=5]

\coordinate (A) at (0, 0);
\coordinate (B) at (1, 0);
\coordinate (C) at (0.5, 0.866);  
\coordinate (AC) at (0.3186,0.5519);
\coordinate (AB) at (0.65,0);
\coordinate (AC1) at (0.1444,0.25); 
\coordinate (CB) at (0.8556,0.25);
\coordinate (AB1) at (0.05,0);
\coordinate (BC) at (0.5341,0.8069);

\coordinate (CB1) at (0.7254,0.4757);
\coordinate (AB2) at (0.44,0);

\coordinate (AB3) at (0.0577,0.1);
\coordinate (BC3) at (0.9422,0.1);

\coordinate (A1) at (0.2, 0.25);
\coordinate (B1) at (0.5, 0.25);
\coordinate (C1) at (0.35, 0.5);

\coordinate (A2) at (0.5, 0.1);
\coordinate (B2) at (0.8, 0.1);
\coordinate (C2) at (0.65, 0.35);

\coordinate (CA3) at (0.4217,0.7304); 
\coordinate (BA4) at (0.86,0);

\draw[thick] (A) -- (B) -- (C) -- cycle;



\coordinate (BC1) at (barycentric cs:A1=1,B1=1,C1=1);  
\coordinate (BC2) at (barycentric cs:A2=1,B2=1,C2=1);  

\filldraw[red] (BC1) circle (0.25pt) node[anchor=south, font=\scriptsize] {BC1};
\filldraw[red] (BC2) circle (0.25pt) node[anchor=north, font=\scriptsize] {BC2};
\filldraw[black] (AB) circle (0pt) node[anchor=north, font=\scriptsize] {};
\filldraw[black] (AC) circle (0pt) node[anchor=north, font=\scriptsize] {};

\filldraw[black] (A) circle (0.2pt) node[anchor=north east, font=\scriptsize] {A};
\filldraw[black] (B) circle (0.2pt) node[anchor=north west, font=\scriptsize] {B};
\filldraw[black] (C) circle (0.2pt) node[anchor=south, font=\scriptsize] {C};

\filldraw[black] (A1) circle (0.3pt) node[anchor=north east, font=\scriptsize] {A1};
\filldraw[black] (B1) circle (0.3pt) node[anchor=north west, font=\scriptsize]{B1};
\filldraw[black] (C1) circle (0.3pt) node[anchor=south, font=\scriptsize] {C1};

\filldraw[black] (A2) circle (0.3pt) node[anchor=north east, font=\scriptsize] {A2};
\filldraw[black] (B2) circle (0.3pt) node[anchor=north west, font=\scriptsize]{B2};
\filldraw[black] (C2) circle (0.3pt) node[anchor=south, font=\scriptsize] {C2};
\end{tikzpicture}
\caption*{(b)}
\end{minipage} 
\hfill 
\begin{minipage}[t]{0.25	 \textwidth}
	\centering

\begin{tikzpicture}[scale=5]

\coordinate (A) at (0, 0);
\coordinate (B) at (1, 0);
\coordinate (C) at (0.5, 0.866);  
\coordinate (AC) at (0.3186,0.5519);
\coordinate (AB) at (0.65,0);
\coordinate (AC1) at (0.1444,0.25); 
\coordinate (CB) at (0.8556,0.25);
\coordinate (AB1) at (0.05,0);
\coordinate (BC) at (0.5341,0.8069);

\coordinate (CB1) at (0.7254,0.4757);
\coordinate (AB2) at (0.44,0);

\coordinate (AB3) at (0.0577,0.1);
\coordinate (BC3) at (0.9422,0.1);

\coordinate (A1) at (0.2, 0.25);
\coordinate (B1) at (0.5, 0.25);
\coordinate (C1) at (0.35, 0.5);

\coordinate (A2) at (0.5, 0.1);
\coordinate (B2) at (0.8, 0.1);
\coordinate (C2) at (0.65, 0.35);

\coordinate (CA3) at (0.4217,0.7304); 
\coordinate (BA4) at (0.86,0);

\draw[thick] (A) -- (B) -- (C) -- cycle;

\draw[thick] (A1) -- (B1) -- (C1) -- cycle;
\draw[thick] (A2) -- (B2) -- (C2) -- cycle;


\coordinate (BC1) at (barycentric cs:A1=1,B1=1,C1=1);  
\coordinate (BC2) at (barycentric cs:A2=1,B2=1,C2=1);  

\filldraw[red] (BC1) circle (0.25pt) node[anchor=south, font=\scriptsize] {BC1};
\filldraw[red] (BC2) circle (0.25pt) node[anchor=north, font=\scriptsize] {BC2};
\filldraw[black] (AB) circle (0pt) node[anchor=north, font=\scriptsize] {};
\filldraw[black] (AC) circle (0pt) node[anchor=north, font=\scriptsize] {};

\filldraw[black] (A) circle (0.2pt) node[anchor=north east, font=\scriptsize] {A};
\filldraw[black] (B) circle (0.2pt) node[anchor=north west, font=\scriptsize] {B};
\filldraw[black] (C) circle (0.2pt) node[anchor=south, font=\scriptsize] {C};

\filldraw[black] (A1) circle (0.3pt) node[anchor=north east, font=\scriptsize] {A1};
\filldraw[black] (B1) circle (0.3pt) node[anchor=north west, font=\scriptsize]{B1};
\filldraw[black] (C1) circle (0.3pt) node[anchor=south, font=\scriptsize] {C1};

\filldraw[black] (A2) circle (0.3pt) node[anchor=north east, font=\scriptsize] {A2};
\filldraw[black] (B2) circle (0.3pt) node[anchor=north west, font=\scriptsize]{B2};
\filldraw[black] (C2) circle (0.3pt) node[anchor=south, font=\scriptsize] {C2};

\end{tikzpicture}
\caption*{(c)}
\end{minipage} 
\begin{minipage}[t]{0.3 \textwidth}
    \centering
\begin{tikzpicture}[scale=5]

\coordinate (A) at (0, 0);
\coordinate (B) at (1, 0);
\coordinate (C) at (0.5, 0.866);  
\coordinate (AC) at (0.3186,0.5519);
\coordinate (AB) at (0.65,0);
\coordinate (AC1) at (0.1444,0.25); 
\coordinate (CB) at (0.8556,0.25);
\coordinate (AB1) at (0.05,0);
\coordinate  (BC) at (0.5341,0.8069);

\coordinate (CB1) at (0.7254,0.4757);
\coordinate (AB2) at (0.44,0);

\coordinate (AB3) at (0.0577,0.1);
\coordinate (BC3) at (0.9422,0.1);

\coordinate (A1) at (0.2, 0.25);
\coordinate (B1) at (0.5, 0.25);
\coordinate (C1) at (0.35, 0.5);

\coordinate (A2) at (0.5, 0.1);
\coordinate (B2) at (0.8, 0.1);
\coordinate (C2) at (0.65, 0.35);

\coordinate (CA3) at (0.4217,0.7304); 
\coordinate (BA4) at (0.86,0);

\draw[thick] (A) -- (B) -- (C) -- cycle;

\draw[thick] (A1) -- (B1) -- (C1) -- cycle;
\draw[thick] (A2) -- (B2) -- (C2) -- cycle;

\draw[thick] (AC) -- (AB);
\draw[thick] (AC1) -- (CB);
\draw[thick] (AB1) -- (BC);

\draw[thick] (CB1) -- (AB2);

\draw[thick] (AB3) -- (BC3);

\draw[thick] (BA4) -- (CA3);

\coordinate (BC1) at (barycentric cs:A1=1,B1=1,C1=1);  
\coordinate (BC2) at (barycentric cs:A2=1,B2=1,C2=1);  

\filldraw[red] (BC1) circle (0.25pt) node[anchor=south, font=\scriptsize] {BC1};
\filldraw[red] (BC2) circle (0.25pt) node[anchor=north, font=\scriptsize] {BC2};
\filldraw[black] (AB) circle (0pt) node[anchor=north, font=\scriptsize] {};
\filldraw[black] (AC) circle (0pt) node[anchor=north, font=\scriptsize] {};

\filldraw[black] (A) circle (0.2pt) node[anchor=north east, font=\scriptsize] {A};
\filldraw[black] (B) circle (0.2pt) node[anchor=north west, font=\scriptsize] {B};
\filldraw[black] (C) circle (0.2pt) node[anchor=south, font=\scriptsize] {C};

\filldraw[black] (A1) circle (0.3pt) node[anchor=north east, font=\scriptsize] {A1};
\filldraw[black] (B1) circle (0.3pt) node[anchor=north west, font=\scriptsize]{B1};
\filldraw[black] (C1) circle (0.3pt) node[anchor=south, font=\scriptsize] {C1};

\filldraw[black] (A2) circle (0.3pt) node[anchor=north east, font=\scriptsize] {A2};
\filldraw[black] (B2) circle (0.3pt) node[anchor=north west, font=\scriptsize]{B2};
\filldraw[black] (C2) circle (0.3pt) node[anchor=south, font=\scriptsize] {C2};
\end{tikzpicture}
\caption*{(d)}
\end{minipage}
\end{figure}
\end{example}

\medskip

\noindent {\bf Step 3.}   From here on, we  fix the map $f$ given in Step 2, along with the multisimplices $X^{ij} \equiv \prod_{n}{X}^{ij}_n$ and $Y^{ij} \equiv \prod_{n}Y^{ij}_n$.  We now describe a  triangulation $\T_n$ of $\S_n$ for each $n$. 

\begin{lemma}\label{lem triangulation T}
	For each $n$, there exists a triangulation $\T_n$ of $\S_n$ such that:
	\begin{enumerate}
		\item For each $i,j$, there exists a simplex $T_n^{ij} \subset X_n^{ij}$ of $\T_n$ with $\s_n^{ij}$ as its barycenter;
		\item  For each $i,j$, letting $T^{ij} \equiv \prod_n T^{ij}_n$, if $\s \notin \cup_{i,j} (T^{ij} \backslash \partial T^{ij})$, there exists $n$ such that for each $\t_n$ that belongs to the carrier of $\s_n$ in $\T_n$,  $\s_n$ does not belong to the carrier of $f_n(\s_{-n}, \t_n)$ in $\T_n$.   
		\item $(\s, \t) \in \Graph(\BR_\e)$ if for each $n$, $\t_n$ is a vertex of $\T_n$ and $f_n(\s)$ belongs to the simplicial neighborhood of the closed star of $\t_n$; 
		\item if $(\s, \t) \in \Graph(\BR_\e)$ where $\s \in T^{ij}$ for some $ij$, then for each $\s' \in T^{ij}$ and each $n$, $G_n(\s_{-n}, \t_n) > \max_{s_n} G_n(\s_{-n}', s_n) - \e$. 
	\end{enumerate}
\end{lemma}


\begin{example}
In Figure \ref{Figlemstep3}, we illustrate condition (2) of Lemma \ref{lem triangulation T}. In the figure, the carrier of $\t_n$ and $\s_n$ is the 2-simplex $A-G-C$, and the carrier of $f_n(\s_{-n}, \t_n)$ is the 2-simplex $A-G-B$.

\begin{figure}[ht]
\caption{Illustration of condition (2) of Lemma \ref{lem triangulation T}}\label{Figlemstep3}
\begin{tikzpicture}[scale=1.5]

\coordinate (A) at (0, 0);
\coordinate (B) at (4, 0);
\coordinate (C) at (2, 3.5);  

\coordinate (G) at (barycentric cs:A=1,B=1,C=1);  

\draw[thick] (A) -- (B) -- (C) -- cycle;

\draw (G) -- (A);
\draw (G) -- (B);
\draw (G) -- (C);

\draw (A) -- (G) -- (B) -- cycle;
\draw (B) -- (G) -- (C) -- cycle;
\draw (C) -- (G) -- (A) -- cycle;

\coordinate (F1) at (barycentric cs:A=1,G=1,C=1); 
\coordinate (F2) at (barycentric cs:A=1,G=1,B=1); 

\coordinate (IntPoint) at (barycentric cs:A=1,G=2,C=2); 

\fill[black] (A) circle (1.5pt) node[anchor=north east] {$A$};
\fill[black] (B) circle (1.5pt) node[anchor=north west] {$B$};
\fill[black] (C) circle (1.5pt) node[anchor=south] {$C$};
\fill[black] (G) circle (1.5pt) node[anchor=south east] {$G$};

\fill[red] (F1) circle (1.5pt) node[anchor=south] {$\t_n$};
\fill[red] (F2) circle (1.5pt) node[anchor=south] {$f_n(\tau_n,\sigma_{-n})$};

\fill[red] (IntPoint) circle (1.5pt) node[anchor=south] {$\sigma_n$};

\end{tikzpicture}

\end{figure}
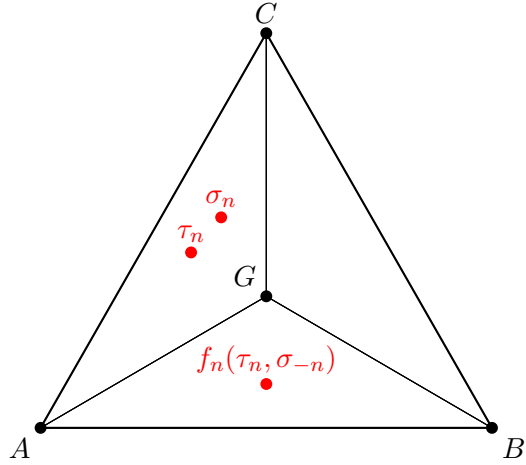

\end{example}

\begin{proof}
There exists $\tilde \eta > 0$ such that: 
\begin{enumerate} 

\item[(5)] if $\s \notin \cup_{i,j} (X^{ij}\backslash \partial X^{ij})$, then there exists $n$ such that $\Vert f_n(\s_{-n}, \t_n) - \s_n \Vert > \tilde \eta$, for all $\t$ within $\tilde \eta$ of $\s$; 
\item[(6)] $(\s, \t) \in \Graph(\BR_\e)$, if $\t$ is within $2\tilde \eta$ of $f(\s)$.  

\end{enumerate}

Let $T_n^{ij}$ be the simplex that has as its vertices  $\{(1-.5\tilde \eta)\s_n^{ij} + .5\tilde \eta v \mid  v \text{ a vertex of } X_n^{ij}\}$. Taking $\tilde \eta>0$ smaller if necessary, we can assume that $T^{ij}_n$ is contained in the interior of the carrier of $\s^{ij}_n$ in $\mathcal{X}_n$ (cf. Lemma \ref{lem function f}) and such that 
\begin{enumerate}
		\item[(7)] for each $(\s, \t) \in \Graph(\BR_\e)$ where $\s \in T^{ij}$, it follows that, for all $n$, $G_n(\s, \t_n) > \max_{s_n} G_n(\s', s_n) -\e$ for all $\s' \in T^{ij}$. 
\end{enumerate}
By construction, $T_n^{ij}$ is contained in $X_n^{ij}$ for each $n$ and $i,j$. Take a triangulation $\T_n'$ of $\S_n$ which refines  $\X_n$ such that:  
\begin{enumerate} 
\item[(8)] $T_n^{ij}$ is a simplex of $\T_n'$ for each $i,j$;  
\medskip
\item[(9)] The diameter of each simplex in $\T'_n$ is no more than $\tilde \eta$. 
\end{enumerate}

Then, properties (1)  (3) and (4) hold for the triangulation $\T_n'$ (from (7), (8) and (9)); for $\s \notin \cup_{i,j}(X^{ij} \backslash \partial X^{ij})$, property (2) also holds for this triangulation because of (4). Since $f_n^{ij}$ is an affine homeomorphism of $X_n^{ij}$ onto $Y_n^{ij}$, both of which have $\s_n^{ij}$ as their barycenters, we have  $f_n^{ij}(\partial  T_n^{ij}) \cap T_n^{ij} = \emptyset$. Therefore, there exists a sufficiently fine subdivision $\T_n$ of $\T_n'$ modulo  $\cup_{i,j} T_n^{ij}$ ---i.e., $\T_n$
subdivides all simplices which are not faces of any $T^{ij}$ (cf. \cite{Z1964})--- such that property (2) holds even for $\s \in \cup_{i,j} (X^{ij} \backslash (T^{ij} \backslash \partial T^{ij}))$. \end{proof}

\begin{example}
The example illustrates how to refine the triangulation shown in Figure \ref{triangmodulo} modulo the 2-simplex $A1-C1-B1$. The point $p$ in the figure is the barycenter of $A1-B1-C1$ and $q$ is another point that does not belong to this $2$-simplex. One can see then in $(b)$ how we can insert vertices away from the 2-simplex $A1-B1-C1$ so as to obtain another triangulation that refines the initial one in $(a)$. The carriers of $p$ and $q$ in the refined triangulation are, in particular, disjoint. Note also that, for the simplices of the refined triangulation that do not share an edge with $A1-B1-C1$, we could continue to subdivide them so as to shrink their diameter to an arbitrarily small number. 
\begin{figure}[ht]
\caption{Refining a triangulation modulo a 2-simplex}\label{triangmodulo}
\begin{minipage}[t]{0.3 \textwidth}
    \centering
\begin{tikzpicture}[scale=1]

\coordinate (A) at (0, 0);
\coordinate (B) at (4, 0);
\coordinate (C) at (2, 3.5);  

\coordinate (G) at (barycentric cs:A=1,B=1,C=1);  

\coordinate (A1) at ($(G)!0.3!(A)$);
\coordinate (B1) at ($(G)!0.3!(B)$);
\coordinate (C1) at ($(G)!0.3!(C)$);

\draw[thick] (A) -- (B) -- (C) -- cycle;

\draw[thick] (A1) -- (B1) -- (C1) -- cycle;

\draw[thick] (C1) -- (C);
\draw[thick] (C1) -- (B);
\draw[thick] (B1) -- (B);
\draw[thick] (B1) -- (A);
\draw[thick] (A1) -- (A);
\draw[thick] (A1) -- (C);

\coordinate (BaryAA1B1) at (barycentric cs:A=1,A1=1,B1=1); 
\coordinate (BaryAA1) at (barycentric cs:A=1,A1=1); 
\coordinate (BaryA1B1) at (barycentric cs:A1=1,B1=1); 
\coordinate (BaryAB1) at (barycentric cs:A=1,B1=1); 


\coordinate (Bary_A_BaryAA1_BaryAA1B1) at (barycentric cs:A=1,BaryAA1=1,BaryAA1B1=1);
\fill[red] (Bary_A_BaryAA1_BaryAA1B1) circle (1pt) node[anchor=north east] {q};

\fill[red] (G) circle (1pt) node[anchor=south] {p};

\fill[black] (A) circle (1pt) node[anchor=north east] {A};
\fill[black] (B) circle (1pt) node[anchor=north west] {B};
\fill[black] (C) circle (1pt) node[anchor=south] {C};

\fill[black] (A1) circle (1pt) node[anchor=south east] {A1};
\fill[black] (B1) circle (1pt) node[anchor=south west] {B1};
\fill[black] (C1) circle (1pt) node[anchor=south] {C1};

\end{tikzpicture}
\caption*{(a)}
\end{minipage}

\begin{tikzpicture}[scale=2]

\coordinate (A) at (0, 0);
\coordinate (B) at (4, 0);
\coordinate (C) at (2, 3.5);  

\coordinate (G) at (barycentric cs:A=1,B=1,C=1);  

\coordinate (A1) at ($(G)!0.3!(A)$);
\coordinate (B1) at ($(G)!0.3!(B)$);
\coordinate (C1) at ($(G)!0.3!(C)$);

\draw[thick] (A) -- (B) -- (C) -- cycle;

\draw[thick] (A1) -- (B1) -- (C1) -- cycle;

\draw[thick] (C1) -- (C);
\draw[thick] (C1) -- (B);
\draw[thick] (B1) -- (B);
\draw[thick] (B1) -- (A);
\draw[thick] (A1) -- (A);
\draw[thick] (A1) -- (C);

\coordinate (BaryAA1B1) at (barycentric cs:A=1,A1=1,B1=1); 
\coordinate (BaryAA1) at (barycentric cs:A=1,A1=1); 
\coordinate (BaryA1B1) at (barycentric cs:A1=1,B1=1); 
\coordinate (BaryAB1) at (barycentric cs:A=1,B1=1); 

\fill[red] (BaryAA1B1) circle (0.7pt) node[anchor=south east] {};
\fill[red] (BaryAA1) circle (0.7pt) node[anchor=south] {};
\fill[red] (BaryAB1) circle (0.7pt) node[anchor=north] {};

\coordinate (Bary_A_BaryAA1_BaryAA1B1) at (barycentric cs:A=1,BaryAA1=1,BaryAA1B1=1);
\fill[red] (Bary_A_BaryAA1_BaryAA1B1) circle (0.7pt) node[anchor=north east] {q};

\fill[red] (G) circle (0.7pt) node[anchor=south] {p};

\fill[black] (A) circle (0.7pt) node[anchor=north east] {A};
\fill[black] (B) circle (0.7pt) node[anchor=north west] {B};
\fill[black] (C) circle (0.7pt) node[anchor=south] {C};

\fill[black] (A1) circle (0.7pt) node[anchor=south east] {A1};
\fill[black] (B1) circle (0.7pt) node[anchor=south west] {B1};
\fill[black] (C1) circle (0.7pt) node[anchor=north] {C1};

\draw[thick] (BaryAA1B1) -- (B1);
\draw[thick] (BaryAA1B1) -- (A1);
\draw[thick] (BaryAA1B1) -- (BaryAB1);
\draw[thick] (BaryAA1B1) -- (A1);
\draw[thick] (BaryAA1) -- (BaryAA1B1);
\draw[thick] (A) -- (BaryAA1B1);

\end{tikzpicture}
\caption*{(b)}
\end{figure}
\end{example}

\medskip

\noindent {\bf Step 4.} We now construct a polytope-form game $\tilde G$ that is equivalent to the original game $G$. The best-reply correspondence of this game is denoted $\tilde{BR}$ and we analogously denote by $\tilde{BR}_{\e}$ the $\e$-best-reply correspondence. For each $n$, let $\tilde S_{n,0}$ be the set of vertices of the triangulation $\T_n$ obtained in Step 3. Let $\tilde \S_{n,0} \equiv \D(\tilde S_{n,0})$ and $\tilde \S_{n,1} \equiv \D(\tilde S_{n,1})$, where $\tilde S_{n,1} = \tilde S_{n+1, 0}$, with the convention that $n+1 = 1$ if $n = N$.  Let $\tilde \S_n = \tilde \S_{n,0} \times \tilde \S_{n, 1}$ and $\tilde \S = \prod_n \tilde \S_n$.  Let $\phi_{n,0}: \tilde \S_{n, 0} \to \S_n$ be the affine function that sends each $\tilde s_{n,0}$ to the corresponding mixed strategy $\s_n \in \S_n$.  We also use $\phi_{n, 1}: \tilde \S_{n,1} \to \S_{n+1}$ to denote the corresponding map for the second factor, $\tilde \S_{n,1}$. The payoffs are now defined as $\tilde G_n(\tilde \s) \equiv G_n({(\phi_{n,0}(\tilde \s_{n,0})}_{n \in \N})$.  In particular, the choice of the strategy $\tilde{\s}_{n,1}$ of any player $n$ is payoff-irrelevant. Therefore, the game $\tilde G$ is equivalent to $G$.

Let $\Graph(\tilde \BR_\e)$ be the graph of the $\e$-best-reply correspondence of $\tilde G$. We now use the map $f$ to construct a map $\tilde f: \tilde \S \to \tilde \S$ such that for each $n$, $\tilde{f}_n$ is independent of $n$'s coordinate $\tilde \s_n$. Let $\tilde \b_{n,0}: \S_n \to \tilde \S_{n,0}$ be the map that assigns to each $\s_n$, the barycentric coordinates of $\s_n$ in the triangulation $\T_n$: that is, $\s_n$ can be written uniquely as a convex combination of the vertices of its carrier in $\T_n$ and $\tilde{\b}_{n,0}$ assigns $\s_n$ to the same convex combination in the face of $\tilde{\S}_{n,0}$ generated by the vertices of that carrier. The map $\tilde \b_{n,1}: \S_{n+1} \to \tilde \S_{n,1}$ is defined similarly. We denote $$\tilde \b_{n,k} \equiv \times_{\tilde s_{n,k}} \beta^{\tilde s_{n,k}}_{n,k},$$ where $k=0,1$, $\beta^{\tilde s_{n,0} }_{n,0}: \S_{n} \to \Re$ assigns the value of the coordinate $\tilde s_{n,0}$ and  $\beta^{\tilde s_{n,1}}_{n,1}: \S_{n+1} \to \Re$ assigns the value of the coordinate $\tilde s_{n,1}$.

\begin{example} 
Figure \ref{barycoordinates} illustrates the definition of the barycentric coordinate map $\tilde \b_{n,k}$. Figure (a) represents a triangulation of a 1-simplex AB, in three vertices \{A, B, C\} and two 1-simplices (AC and BC). We singled-out a point in blue which has $BC$ as a carrier. We define the barycentric coordinate map $b$ from the 1-simplex $AB$ to a 2-simplex. First, the vertices of the 2-simplex are images under $b$ of vertices of the triangulation, as indicated in figure (b). In each 1-simplex of the triangulation represented in (a), the barycentric coordinate map is affine over the same colored face in (b). The barycentric coordinates of the point in blue depicted in (a) are its coordinates with respect to the simplex that carries it in the triangulation, so its coordinates are $(1/2,1/2)$---where the first entry corresponds to vertex $C$ and the second to $B$. These coordinates are preserved in figure (b), since $b$ is affine from the red simplex on (a) over the red simplex on (b). Note in addition that the image of $b$ is the space of a subcomplex of the face complex (cf. Example \ref{examplefacecomplex}) of (b): the subcomplex is defined by the green and red 1-simplices and their vertices, and the space is their union.

\begin{figure}[ht]
\caption{Illustration of $\tilde \b_{n,k}$}\label{barycoordinates}
\medskip 
\centering
\begin{minipage}[t]{0.3 \textwidth}

\begin{tikzpicture}[scale =2]
    \draw[thick] (0,1) -- (1,0);
	\draw[thick, red] (0.5,0.5) -- (1,0);
	\draw[thick, green] (0.5,0.5) -- (0,1);
    
    \filldraw[black] (0,1) circle (1pt) node[left] {\small A=(0,1)};
    \filldraw[black] (1,0) circle (1pt) node[right] {\small B=(1,0)};
    
    \filldraw[black] (0.5,0.5) circle (1pt) node[above right] {\small C=(1/2, 1/2)};
    
    \filldraw[blue] (0.75,0.25) circle (1pt) node[above right] {\small (1/2)C + (1/2)B};
    
\end{tikzpicture}
\caption*{(a)}
\end{minipage} 
\begin{minipage}[t]{0.6 \textwidth}
	\centering
\begin{tikzpicture}
    \draw[thick] (0,0) -- (4,0) -- (2,3.46) -- cycle;
	\draw[thick, red] (4,0) -- (2,3.46);
	\draw[thick, green] (2,3.46) -- (0,0);    

    \filldraw[black] (0,0) circle (2pt) node[below left] {\small $b(A) = (0,0,1)$};
    \filldraw[black] (4,0) circle (2pt) node[below left] {\small $b(B) = (1,0,0)$};
    \filldraw[black] (2,3.46) circle (2pt) node[above] {\small $b(C) = (0,1,0)$};
    
    \filldraw[blue] (3,1.73) circle (2pt) node[left] {\small $b((1/2) C + (1/2) B) = (1/2)b(C) + (1/2)b(B)$};
    
    \draw[dashed] (4,0) -- (3,1.73); 
    
\end{tikzpicture}
\caption*{(b)}
\end{minipage} 

\end{figure}
\end{example}

Define $\tilde f$ by $\tilde f_n(\tilde \s) = (\tilde f_{n,0}(\tilde \s), \tilde f_{n,1}(\tilde \s))$, where: 
\[
\tilde f_{n,0}(\tilde \s) \equiv \tilde \beta_{n,0}(f_n(\phi_{-n, 0}(\tilde \s_{-n,0}), \phi_{n-1, 1}(\tilde \s_{n-1,1})))
\]
with the convention that $n-1 = N$ if $n = 1$, and
\[
\tilde f_{n,1}(\tilde \s) \equiv \tilde \beta_{n,1}(\phi_{n+1,0}(\tilde \s_{n+1,0})).
\]

Let $\tilde \S_{n,k}^* \equiv \tilde \beta_{n,k}(\S_{n+k})$, for $k = 0, 1$. The set $\tilde \S_{n,k}^*$ is then the space of a subcomplex $\mathcal{S}^*_{n,k}$ of the face complex of $\D(\tilde{S}_{n,k})$. For each $n$ and $k = 0,1$, let $\tilde T_{n,k}^{ij} \equiv \tilde \beta_{n,k}({(f_{n+k}^{ij})}^{-1}(T_{n,k}^{ij}))$. Recall that $f^{ij}_n: X^{ij}_n \to Y^{ij}_n, X^{ij}_n \subset \text{int}(Y^{ij}_n)$ is an affine homeomorphism.  Since $T^{ij}_{n,k} \subseteq \text{int}(X^{ij}_n)$, then, for $k \in \{0,1\}$, $(f^{ij}_n)^{-1}(T^{ij}_{n,k}) \subseteq T^{ij}_{n,k}$. Hence, $\tilde \beta_{n,k}(T_{n,k}^{ij})$ is a face of the simplex $\tilde{\S}_{n,k}$ and $\tilde{T}^{ij}_{n,k}$ is contained in that face. 


Let

$$\tilde \S_n^* \equiv \tilde \S_{n,0}^* \times \tilde \S_{n,1}^*  \hspace*{0.5cm} ; \hspace*{0.5cm} \tilde T_n^{ij} \equiv \tilde T_{n,0}^{ij} \times \tilde T_{n,1}^{ij}  \hspace*{0.5cm} ; \hspace*{0.5cm} \tilde \S^* \equiv \prod_n \tilde \S^*_n  \hspace*{0.5cm} ; \hspace*{0.5cm} \tilde T^{ij} \equiv \prod_n \tilde T_n^{ij}.$$

For each $i,j,n$ and $k$, let $\tilde \s_{n,k}^{ij} \equiv \tilde{\beta}_{n,k}(\s_{n+k}^{ij})$, $\tilde \s_n^{ij} \equiv (\tilde \s_{n,0}^{ij}, \tilde \s_{n,1}^{ij})$ and $\tilde \s^{ij} \equiv {(\tilde \s_n^{ij})}_{n \in \N}$. By Property (1) of Lemma \ref{lem triangulation T}, for each $n$, $i, j$ and $k = 0, 1$,  $\tilde \s_{n,k}^{ij}$ is the barycenter of both $\tilde T_n^{ij}$ and $\tilde \beta_{n,k}(T_{n,k}^{ij})$.

The following lemma summarizes key properties of $\tilde f$.

\begin{lemma}\label{lem function f tilde}
	The function $\tilde f$ maps $\tilde \S$ into $\tilde \S^*$ and player $n$'s coordinate functions $\tilde{f}_{n,0}$ and $\tilde{f}_{n,1}$ are independent of the coordinate $\tilde \s_n$.  Moreover,
	\begin{enumerate}
		\item $\tilde f_{n,1}$ is a function of $\tilde \s_{n+1, 0}$ and the restriction of the function to $\tilde \S_{n+1, 0}^*$ is affine on each simplex of $\mathcal{S}_{n+1, 0}^*$;
		\item the restriction of $\tilde f_n$ to each $\tilde T^{ij}$ is the cartesian product of affine homeomorphisms $\tilde f_{n,0}: \tilde T_{n-1, 1}^{ij} \to \tilde \beta_{n,0}(T_n^{ij})$ and $\tilde f_{n,1}: \tilde T_{n+1, 0}^{ij} \to  \tilde \beta_{n, 1}(T_{n+1}^{ij})$;
		\item if for each $n \in \N, k=0,1$, $\tilde f_{n,k}(\tilde \s)$ belongs to the simplicial neighborhood of the closed star of a pure strategy $\tilde s_{n,k} \in \tilde \S^*_{n,k}$ in the triangulation $ \mathcal{S}^*_{n,k}$, then $(\tilde \s, \tilde s) \in \Graph(\tilde \BR_\e)$;
		\item if $\tilde \s \notin \tilde T^{ij} \backslash \partial \tilde T^{ij}$ for any $i,j$ and if for all $n$, $\phi_{n-1, 1}(\tilde \s_{n-1, 1})$ belongs to the carrier of $\phi_{n,0}(\tilde \s_{n,0})$ in $\T_n$, then there exists $n$ such that $\tilde \s_{n,0}$ does not belong to the carrier of $\tilde f_{n,0}(\tilde \s)$ in $\mathcal{S}^*_{n,0}$; 
		\item the only fixed points of $\tilde f$ are the profiles $\tilde \s^{ij}$ and each $\tilde \s^{ij}$ has index $r^{ij}$.
	\end{enumerate}
\end{lemma}

\begin{proof}It follows from the definition of $\tilde f$ that each player's coordinate function $\tilde f_n$ is independent of his strategies. By definition of $\tilde{f}$ it maps into $\tilde \S^*$.  The numbered items (1), and (2) also follow by definition of $\tilde{f}$ and $\tilde \b_{n,k}$.  As for (3), it follows from property (3) of Lemma \ref{lem triangulation T}.  Property (4) follows from property (2) of Lemma \ref{lem triangulation T} in the case where ${(\phi_{n,0}(\tilde \s_{n,0}))}_{n \in \N} \notin \cup_{i,j} (T^{ij} \backslash \partial T^{ij})$.  Consider now the case where ${(\phi_{n,0}(\tilde \s_{n,0}))}_{n \in \N}$ belongs to the set $T^{ij} \backslash (\partial T^{ij} \cup ((f^{ij})^{-1}(T^{ij}\backslash \partial T^{ij}))$ for some $i,j$. There exists then $n$ such that $\phi_{n,0}(\tilde \s_{n,0})$ belongs to the set $T_n^{ij} \backslash (\partial T_n^{ij} \cup ((f^{ij}_n)^{-1}(T_n^{ij} \backslash \partial T_n^{ij}))$. As $\phi_{n-1, 1}(\tilde \s_{n-1,1})$ belongs to the same simplex as $\phi_{n,0}(\tilde \s_{n,0})$, it also lies in  $T_n^{ij} \backslash ((f^{ij}_n)^{-1}(T_n^{ij} \backslash \partial T_n^{ij}))$. Its image under $f_{n}^{ij}$ (recall it depends only on $\phi_{n-1, 1}(\tilde \s_{n-1,1})$) lies outside $T_n^{ij} \backslash \partial T_n^{ij}$. This completes the proof of property (4).  
	
	There remains to prove point (5). The index of $\tilde \s^{ij}$ can be computed using the restriction of $\tilde f$ to $\tilde T^{ij}$. Call this restriction $\hat{f}$ and note that each coordinate function $\hat f_n$ is a cartesian product of affine homeomorphisms: $\hat f_n = \tilde f_{n,0} \times \tilde f_{n-1, 1}: \tilde T_{n,0}^{ij} \times \tilde T_{n-1, 1}^{ij} \to \tilde  \beta_{n,0}(T_n^{ij}) \times \tilde \beta_{n-1, 1}(T_n^{ij})$, where $\hat f_n(\tilde \s_{n,0}, \tilde \s_{n-1, 1}) = (\tilde \beta_{n,0}(f_n^{ij}(\phi_{n-1,1}(\tilde \s_{n-1, 1}))), \tilde \s_{n,0})$.  For $\l \in [0, 1]$, let $\hat f^\l$ be the map sending $(\tilde \s_{n,0}, \tilde \s_{n-1, 1})$ to $(\tilde \beta_{n,0}(f_n^{ij}((1-\l)\phi_{n,0}(\tilde \s_{n,0}) + \l \phi_{n-1,1}(\tilde \s_{n-1,1}))), (1-\l)\tilde \s_{n-1,1} + \l \tilde \s_{n,0})$.  When $\l = 1$, we get the map $\hat f$.  For $\l > 0$, the unique fixed point is $(\tilde \s_n^{ij}, \tilde \s_{n-1, 1}^{ij})$ and for $\l = 0$, we have a component of fixed points $\{\, \tilde \s_{n,0}^{ij} \, \} \times \tilde T_{n,1}^{ij}$.  By the multiplication property of the index, the index of the component of fixed points for $\l = 0$ is the index of $\s_n^{ij}$ under the map $f_n^{ij}$.  Therefore, the index of $\tilde \s^{ij}$ is $r^{ij}$, as claimed.\end{proof}

\medskip

\noindent {\bf Step 5.} We now construct a family of $\e$-perturbed games, indexed by $\tilde \s \in \tilde \S$.  First we need to introduce some notation.  Given vectors $\tilde g^0 \in \prod_n \Re^{\tilde S_{n,0}}$ and $\tilde g^1 \in \prod_n \Re^{\tilde S_{n,1}}$, we define the game $\tilde G \oplus \tilde g^0 \oplus \tilde g^1$ as a polytope-form game where for each $\tilde \s$ and each $n$, the payoff is $\tilde G_n(\tilde \s) + \tilde g_n^0 \cdot \tilde \s_{n,0} + \tilde g_n^1 \cdot \tilde \s_{n,1}$ (where $a \cdot b$ denotes the scalar product of the vectors $a$ and $b$). Thus $\tilde g_{n,\tilde s_{n,0}}^0$ represents a ``bonus'' for playing $\tilde s_{n,0}$ and similarly $\tilde g_{n,\tilde s_{n,1}}^0$, for each $\tilde s_{n,1}$.  Given functions $\tilde g^0: \tilde \S \to \prod_n \Re_+^{\tilde S_{n,0}}$, $\tilde g^1: \tilde \S \to \prod_n \Re_+^{\tilde S_{n,1}}$ and profile $\tilde \s$, we have a finite polytope-form game $\tilde G \oplus \tilde g^0(\tilde \s) \oplus \tilde g^1(\tilde \s)$.   


For the time being fix $\e_0 > 0$. The exact choice will be determined later.  
We first define $\tilde g^1(\cdot)$.  For each $n$ and $\tilde s_{n,1}$, it depends only on $\tilde \s_{n+1, 0}$ and it is given by 
\[
\tilde g_{n, \tilde s_{n,1}}^1(\tilde \s_{-n}) = \e_0\tilde \beta^{\tilde{s}_{n,1}}_{n, 1}(\phi_{n+1, 0}(\tilde \s_{n+1, 0})).
\]
where $\tilde \beta^{\tilde s_{n,1}}_{n,1}$ gives the $\tilde s_{n,1}$ coordinate of the barycentric coordinates w.r.t. the triangulation $\T_{n+1}$.
To define $\tilde g^0$, we first construct for each $n$, a function  $r_n : \S_{-n} \to \Re$ satisfying:	
	\begin{enumerate}
		\item For all $\s_{-n} \in \S_{-n}, r_n(\s_{-n}) \ge \text{max}_{s_n}G_n(s_n, \s_{-n})$;
		\item The map $r_n$ is constant over $T^{ij}_{-n}$.
		\item If $\t_n$ is an $\e$-best reply against $\s_{-n}$, then $r_n(\s_{-n}) - G_n(\t_n, \s_{-n}) < \e$.
	\end{enumerate}
To construct this function, for each $i,j$ and $\s \in T^{ij}$, define $r_n(\s) = \max_{s_n \in S_n, \s' \in T^{ij}} G_n(\s', s_n)$.  On the sets $T^{ij}$ the function satisfies the three desired properties (thanks to point (4) of Lemma \ref{lem triangulation T}). There exists an open neighborhood $W^{ij}$ of $T^{ij}$ such that (constantly) extending this function to $W^{ij}$ still satisfies the three properties.  By Urysohn's Lemma,  there is a continuous function $\xi: \S \to [0, 1]$ that is one on $\cup_{ij} T^{ij}$ and zero outside $\cup_{ij} W^{ij}$.  The function  $r_n$ is then obtained as follows:  $r_n(\s) = \text{max}_{s_n}G_n(s_n, \s_{-n})$ if $\s \notin W^{ij}$; $r_n(\s) = \xi(\s)\max_{s_n \in S_n, \s' \in T^{ij}} G_n(\s', s_n) + (1-\xi(\s))\max_{s_n}G_n(s_n, \s_{-n})$ if $\s \in W^{ij}$.     

Recall that $\tilde \S_{n,0}^*$ is the space of a subcomplex of the face complex of $\tilde \S_{n,0}$.  Let $\tilde \g_{n,s_{n,0}}: \tilde \S_{n,0}^* \to [0, 1]$ be a Urysohn function that is one on the closed star of $\tilde s_{n,0}$ in the complex $\mathcal{S}_{n,0}^*$ and zero outside the simplicial neighborhood of the closed star of $\tilde s_{n,0}$ in this complex. Then,
\[
\tilde g_{n, \tilde s_{n,0}}^0(\tilde \s_{-n}) \equiv \tilde \g_{n, \tilde s_{n,0}}(\t_n)[r_n(\phi_{-n, 0}(\tilde \s_{-n,0})) - G_n(\tilde \s_{-n, 0}, \phi_{n,0}(\tilde{s}_{n,0}))]  +  \e_0 \tilde f_{n, \tilde s_{n,0}}(\tilde \s_{-n}),
\]
where
\[
\t_n = \tilde \beta_{n,0}(f_{n}(\phi_{-n, 0}(\tilde \s_{-n}), \phi_{n-1, 1}(\tilde \s_{n-1, 1}))),
\]
and $r_n(\cdot)$ is the function we just defined.  The following lemma sets out key properties of the perturbations.  

\begin{lemma}\label{lem affine}
	Recall that $\e>0$ is fixed from the beginning of the proof (cf. step 2). Each player $n$'s coordinate functions $\tilde g_n^0$ and $\tilde g_n^1$ are independent of $\tilde \s_n$.  Moreover:
	\begin{enumerate}
		\item the restriction of $\tilde g_{n, \tilde s_{n,1}}^1$ to $\S_{n+1, 0}^*$ is affine on each simplex of $\mathcal{S}_{n+1, 0}^*$;
		\item the restriction of $\tilde g^0_n$ to each $\tilde T^{ij}$ is a cartesian product  $\prod_{\tilde{s}_{n,0}}\tilde g_{n,\tilde s_{n,0}}^{ij,0}$, where for each $\tilde s_{n,0} \in \tilde T^{ij}_n$, $\tilde g_{n,\tilde s_{n,0}}^{ij,0} :\tilde T_{-n,0}^{ij} \times \tilde T_{n-1, 1}^{ij} \to \Re_+$ and $\tilde G_n(\tilde{s}_{n,0}, \tilde{\s}_{-n}) + \tilde g_{n,\tilde s_{n,0}}^{ij,0}(\tilde \s_{-n}) = r_n(\phi_{-n,0}(\tilde \s_{-n, 0}))  + \e_0 \tilde f_{n,\tilde{s}_{n,0}}(\tilde \s_{n-1,1})$, which is an affine function. 
		\item $\Vert \tilde g^0 \Vert + \Vert \tilde g^1 \Vert < \e$, if $\e_0$ is small.
	\end{enumerate}
\end{lemma}

\begin{proof} Item (1) follows from the fact that $\tilde \b^{\tilde{s}_{n,1}}_{n,1}$ is affine in each cell of $\mathcal{S}^{*}_{n+1,0}$ and $\phi_{n+1,0}$ is affine by definition.  For (2), the map $\tilde g^0_n$ is a cartesian product $\prod_{\tilde{s}_{n,0}}\tilde g_{n,\tilde s_{n,0}}^{ij,0}$ by definition. Fix  $\tilde \s_{-n} \in \tilde T^{ij}_{-n}$. Note that any vertex $\tilde{s}_{n,0}$ of $\tilde T^{ij}_n$ is such that $\tilde \g_{n, \tilde s_{n,0}}(\t_n) =1$, where $\t_n = \tilde \beta_{n,0}(f_{n}(\phi_{n-1, 1}(\tilde \s_{n-1, 1})))$ (recall $f_n$ depends exclusively on $\s_{n-1,1}$ in this region). Therefore, $\tilde G_n(\tilde{s}_{n,0}, \tilde{\s}_{-n}) + \tilde g_{n,\tilde s_{n,0}}^{ij,0}(\tilde \s_{-n}) = r_n(\phi_{-n,0}(\tilde \s_{-n, 0}))  + \e_0 \tilde f_{n,\tilde{s}_{n,0}}(\tilde \s_{n-1,1})$. By definition $r_n(\phi_{-n,0}(\tilde \s_{-n, 0}))$ is constant on $\tilde T^{ij}$, which makes $r(\phi_{-n,0}(\tilde \s_{-n, 0}))  + \e_0 \tilde f_{n,\tilde{s}_{n,0}}(\tilde \s_{n-1,1})$ affine. For (3), it is sufficient to prove that the difference 

$$\tilde \g_{n, \tilde s_{n,0}}(\t_n)[r_n(\phi_{-n, 0}(\s_{-n,0})) - G_n(\tilde \s_{-n, 0}, \phi_{n,0}(\tilde{s}_{n,0}))]$$ is strictly less than $\e>0$. First, if $\t_n$ is outside the simplicial neighborhood of the closed star of $\tilde s_{n,0}$, then it follows that $\tilde \g_{n, \tilde s_{n,0}}(\t_n) =0$ and the result follows. If $\t_n$ is inside that simplicial neighborhood, then item (3) of Lemma \ref{lem function f tilde} implies $\tilde s_{n,0}$ is an $\e$-best reply to $\tilde \s$ in $\tilde G$. Hence $\e > r_n(\phi_{-n, 0}(\s_{-n,0})) - G_n(\tilde \s_{-n, 0}, \phi_{n,0}(\tilde{s}_{n,0})) \ge 0$, which implies the desired result.\end{proof}

We conclude this step with a lemma concerning best replies in these perturbed games. Its proof, too, follows from the construction of the perturbed games.

\begin{lemma}\label{lem best replies in perturbed games}
	The profile $(\tilde s_{n,0}, \tilde s_{n,1})$ is a best reply to $\tilde \s$ in the game $\tilde G \oplus \tilde g^0(\tilde \s) \oplus \tilde g^1(\tilde \s)$ only if:
	\begin{enumerate}
		\item $\tilde s_{n,0}$ is a vertex of the carrier of $\tilde f_{n,0}(\tilde \s_{-n,0}, \tilde \s_{n-1, 1})$ in $\mathcal{S}_{n,0}^*$;
		\item $\tilde s_{n,1}$ is a vertex of the carrier of  $\tilde f_{n,1}(\tilde \s_{n+1,0})$ in $\mathcal{S}_{n+1,0}^*$.
	\end{enumerate} 
\end{lemma}

\begin{proof} 
To prove (1), take $(\tilde t_{n,0}, \tilde s_{n,1})$ such that $\tilde t_{n,0}$ belongs to the carrier of $\tilde f_{n,0}(\tilde \s_{-n,0}, \tilde \s_{n-1, 1})$ in $\mathcal{S}_{n,0}^*$.  Then the payoff it obtains  in $\tilde G \oplus \tilde g^0(\tilde \s) \oplus \tilde g^1(\tilde \s)$  is $r_n(\phi_{-n,0}(\tilde \s_{-n, 0}))  + \e_0 \tilde f_{n,\tilde{t}_{n,0}}(\tilde \s_{-n}) + \tilde g_{n, \tilde s_{n,1}}^1(\tilde \s_{-n})$; this quantity is strictly larger than the payoff from $(\tilde s_{n,0}, \tilde s_{n,1})$ where  $\tilde s_{n,0}$ is not in that carrier, because $\g_{n, \tilde s_{n,0}}(\t_n) \le 1$ and $\tilde f_{n,\tilde{s}_{n,0}}(\tilde \s_{-n}) = 0$. The proof of (2) follows a similar reasoning: since $\tilde g_{n, \tilde s_{n,1}}^1(\tilde \s_{-n}) = \e_0\tilde \beta^{\tilde{s}_{n,1}}_{n, 1}(\phi_{n+1, 0}(\tilde \s_{n+1, 0})) = \e_0 \tilde f_{\tilde{s}_{n,1}}(\tilde \s_{n+1,0})$, a best reply $\tilde s_{n,1}$ must be a vertex of the of the carrier $\tilde f_{n,1}(\tilde \s_{n+1,0})$.\end{proof} 

\medskip

\noindent{\bf Step 6.} 
Let $\tilde \varphi: \tilde \S \twoheadrightarrow \tilde \S$ be the correspondence that assigns to each $\tilde \s$ the set of best replies to it in the polytope-form game $\tilde G \oplus \tilde g^0(\tilde \s) \oplus \tilde g^1(\tilde \s)$. The correspondence $\tilde \varphi$ is non-empty, compact, and convex-valued. It is also upper semi-continuous. Hence it has a fixed point and a well-defined index for its fixed points.  We now characterize the fixed points.

\begin{lemma}\label{lem phi corresp}
	The fixed points of $\tilde \varphi$ are the profiles $\tilde \s^{ij}$. The index of each $\tilde \s^{ij}$ is $r^{ij}$. 
\end{lemma}

\begin{proof}To prove this step we will show that for all $\l \in [0, 1]$ the only fixed points of $\l \tilde f + (1-\l)\tilde \varphi$ are the $\tilde \s^{ij}$'s. The proof then follows from item (5) of Lemma \ref{lem function f tilde}.

Take some $\tilde \s$ that is a fixed point of $\l \tilde f + (1-\l) \tilde \varphi$ for some $\l \in [0, 1]$.  It follows from (1) of Lemma \ref{lem best replies in perturbed games}  that $\tilde \s_{n,0}$ belongs to the simplex of $\mathcal{S}_{n,0}^*$ that contains $\tilde f_{n,0}(\tilde \s_{-n,0}, \tilde \s_{n-1,1})$ in its interior. An optimal $\tilde s_{n-1,1}$ has to be a vertex of the simplex containing $\tilde \s_{n,0}$ (which is also the carrier of the image of $\tilde \s_{n,0}$ under $\tilde f_{n-1,1}$). Hence $\tilde \s_{n-1, 1}$ belongs to the simplex containg $\tilde \s_{n,0}$.  By property (4) of Lemma \ref{lem function f tilde},  $\tilde \s$ must be in $\tilde T^{ij}$ for some $i,j$. It then follows in this case $\tilde \s = \tilde \s^{ij}$.   \end{proof}

\medskip

Recall that for each $\tilde \s \in \tilde \S$, we can define a polytope-form game $\tilde G \oplus \tilde g^0(\tilde \s) \oplus \tilde g^{1}(\tilde \s)$. Hence we have a continuous family of polytope-form games indexed by $\tilde \s \in \tilde \S$: $\{\tilde G \oplus \tilde g^0(\tilde \s) \oplus \tilde g^{1}(\tilde \s)\}_{\tilde \s \in \tilde \S}$. We will now use this family to construct a game $\hat{G}$ that is equivalent to $\tilde{G}$ and identify a finite payoff perturbation that attains the objectives of our main theorem. This construction is performed in the final three steps. 

\medskip

\noindent{\bf Step 7.} In this step, we describe a triangulation of $\tilde \S_n$ for each $n$ that allows us to convert the continuous family of perturbed games to a single perturbed game.
We first observe that we can obtain a $\xi > 0$ such that: 
\begin{enumerate} 

	\item[(A)] for $\tilde \s \notin \cup_{i,j} (\tilde T^{ij} \backslash \partial \tilde T^{ij})$, if $\tilde g^0 \in \Re^{\sum_{n}|\tilde S_{n,0}|}$ and $\tilde g^1 \in \Re^{\sum_{n}|\tilde S_{n,1}|}$ satisfy $\Vert \tilde g^0 - \tilde g^0(\tilde \s) \Vert  + \Vert \tilde g^1 - \tilde g^1(\tilde \s) \Vert < 2\xi$, then $\tilde \s$ is not an equilibrium of $\tilde G \oplus \tilde g^0 \oplus \tilde g^1$.
	\medskip
	
\end{enumerate}
\medskip

Claim (A) follows from the fact that if $\tilde \s \notin \cup_{i,j} (\tilde T^{ij} \backslash \partial \tilde T^{ij})$, then $\tilde \s$ is not an equilibrium of $\tilde G \oplus \tilde g^0(\tilde \s) \oplus \tilde g^1(\tilde \s)$ (cf. Lemma \ref{lem phi corresp}). Therefore, for sufficiently small perturbations $\tilde g^0$ of $\tilde g^0(\tilde \s)$ and  $\tilde g^1$ of $\tilde g^1(\tilde \s)$, $\tilde \s$ is also not an equilibrium of the finite game $\tilde G \oplus \tilde g^0 \oplus \tilde g^1$.

\medskip

Since $g(\cdot)$ is uniformly continuous, there exists $\zeta > 0$ such that $\Vert g(\tilde \s) - g(\tilde \s')\Vert  < \xi$ if $\Vert \tilde \s - \tilde \s' \Vert < \zeta$.  
Recall that $\tilde S_{n,0}$ is the set of vertices of $\T_n$ and $\tilde \S_{n,0} = \tilde \S_{n-1, 1} = \D(\tilde S_{n,0})$. Therefore, any triangulation of $\D(\tilde S_{n,0})$ induces a corresponding triangulation of $\tilde \S_{n,0}$ and $\tilde \S_{n-1, 1}$. 

\begin{lemma}\label{lem triangulation T tilde}
	For each $n$, there exists a triangulation $\tilde \T_n$ of $\D(\tilde S_{n,0})$ with the following properties. 
	\begin{enumerate}
		\item the diameter of each simplex is less than $\zeta$;
		\item for each $i,j$, $\tilde T_{n,0}^{ij}$ is the space of a subcomplex of $\T_n$,  $\tilde \s^{ij}_{n,k}$  lies in the interior of a simplex $\hat T_{n,k}^{ij} \in \mathcal{T}_n$ for $k= 0,1$, with dim($\hat{T}_{n,k}^{ij}$) = dim($\S_{n+k}$), and $(\tilde \s^{ij}_{n,0}, \tilde \s^{ij}_{n,1})$ is not in the convex hull of $(\text{dim}(\S_n) + \text{dim}(\S_{n+1})$) vertices of $\hat T_{n,0}^{ij} \times \hat T_{n,1}^{ij}$;
		\item there is a convex piecewise linear function $\g_n: \D(\tilde S_{n,0}) \to \Re_{+}$ such that:
		\begin{enumerate}
			\item the map is linear precisely on the simplices of $\tilde \T_n$;
			\item the function is constant over each $\hat T_{n,0}^{ij}$.
		\end{enumerate}
	\end{enumerate}
\end{lemma}

\begin{proof} 
	See Appendix. 
\end{proof}

\begin{example} 
Figure \ref{lemtriangT} illustrates condition (2) of Lemma \ref{lem triangulation T tilde}. The idea is to view AD as the 1-simplex $\hat T_{n,0}^{ij}$ and AB as the 1-simplex $\hat T_{n,1}^{ij}$. The condition says the point $\s^{ij}_n$ should avoid both diagonals AC and DB (as figure (b) shows). The purpose of this requirement will be clear in step 8 and it is illustrated in Example \ref{expolytopecomponent}.

\begin{figure} 
\caption{Illustration of property (2) of Lemma \ref{lem triangulation T tilde}}\label{lemtriangT}
\medskip
\begin{minipage}[t]{0.4 \textwidth}
	\centering
\begin{tikzpicture}[scale=2]

\coordinate (A) at (0, 0);
\coordinate (B) at (2, 0);
\coordinate (C) at (2, 2);
\coordinate (D) at (0, 2);

\draw[thick] (A) -- (B) -- (C) -- (D) -- cycle;

\filldraw[black] (A) circle (1pt) node[anchor=north east] {A};
\filldraw[black] (B) circle (1pt) node[anchor=north west] {B};
\filldraw[black] (C) circle (1pt) node[anchor=south west] {C};
\filldraw[black] (D) circle (1pt) node[anchor=south east] {D};


\coordinate (P) at (1.2, 1.5); 

\filldraw[red] (P) circle (1pt) node[anchor=south east] {$\tilde{\sigma}^{ij}_n$};

\coordinate (ProjAB) at (1.2, 0);  
\draw[dashed] (P) -- (ProjAB);  
\filldraw[blue] (ProjAB) circle (1pt) node[anchor=north] {$\tilde{\sigma}^{ij}_{n,1}$};  

\coordinate (ProjAD) at (0, 1.5);  
\draw[dashed] (P) -- (ProjAD);  
\filldraw[blue] (ProjAD) circle (1pt) node[anchor=east] {$\tilde{\sigma}^{ij}_{n,0}$};  

\end{tikzpicture}
\caption*{(a)}
\end{minipage} 
\begin{minipage}[t]{0.4 \textwidth}
\centering
\begin{tikzpicture}[scale=2]

\coordinate (A) at (0, 0);
\coordinate (B) at (2, 0);
\coordinate (C) at (2, 2);
\coordinate (D) at (0, 2);

\draw[thick] (A) -- (B) -- (C) -- (D) -- cycle;

\filldraw[black] (A) circle (1pt) node[anchor=north east] {A};
\filldraw[black] (B) circle (1pt) node[anchor=north west] {B};
\filldraw[black] (C) circle (1pt) node[anchor=south west] {C};
\filldraw[black] (D) circle (1pt) node[anchor=south east] {D};

\draw[dashed] (A) -- (C);  
\draw[dashed] (B) -- (D);  

\coordinate (P) at (1.2, 1.5); 

\filldraw[red] (P) circle (1pt) node[anchor=south east] {$\tilde{\sigma}^{ij}_n$};

\end{tikzpicture}
\caption*{(b)}
\end{minipage}

\end{figure}
\end{example}

\medskip

\noindent{\bf Step 8.} We now define a game $\hat G$ in normal form.   Let $\hat S_{n,0}$ be the set of vertices of $\tilde \T_n$.  For each $n$, let  $\hat S_{n,1} \equiv  \hat S_{n+1, 0}$ and  $\hat S_n \equiv \hat S_{n,0} \times \hat S_{n,1}$.  The pure strategy set of player $n$ in $\hat G$ is $\hat S_n$.  Let $\hat \S_{n, 0} = \D(\hat S_{n,0})$,  $\hat \S_{n,1}= \D(\hat S_{n, 1})$, and $\hat \S = \hat \S_{n,0} \times \hat \S_{n,1}$. For $k = 0, 1$, let $\hat \mu_{n,k}: \hat \S_n \to  \hat \S_{n,k}$ be the map that sends each $\hat \s_n$ to the marginal distribution over $\hat S_{n,k}$.  Each $\hat s_{n,0} \in \hat S_{n,0}$ corresponds to a mixed strategy in $\tilde \S_{n,0}$.  Therefore, there exists for $k = 0, 1$, an affine map $\hat \phi_{n,k}: \hat \S_{n,k} \to \tilde \S_{n,k}$. Let $\hat{\phi}_{n} \equiv \hat \phi_{n,0} \times \hat \phi_{n,1}$ and $\hat{\mu}_{n} \equiv \hat \mu_{n,0} \times \hat \mu_{n,1}$. The payoffs in $\hat G$ are defined, for each player $n$, by $\hat G_n(\hat \s) = \tilde G_n(\hat \phi_0 \circ \hat \mu_0 (\hat \s))$, where $\hat{\phi}_0 \equiv \times_{n} \hat{\phi}_{n,0}$ and $\hat \mu_{0} \equiv \times_{n} \hat \mu_{n,0}$. The game $\hat G$ just defined is equivalent to $\tilde G$ and hence to $G$.  

For each $n$, $i,j$ and $k = 0, 1$, let $\hat S_{n,k}^{ij}$ be the set of vertices of $\hat T_{n,k}^{ij}$ obtained in (2) of Lemma \ref{lem triangulation T tilde} and let $\hat \S_{n,k}^{ij}$ be the face of $\hat \S_{n,k}$ consisting of strategies with support contained in $\hat S_{n,k}^{ij}$.  Let $\hat \S_n^{ij}$ be the set of mixed strategies $\hat \s_n$ such that $\hat \mu_{n,k}(\hat \s_n) \in \hat \S_{n,k}^{ij}$ for each $k$.  Let $\hat \S^{ij} \equiv \prod_n \hat \S_n^{ij}$. Let $\hat C^{ij} = \{ \hat \s \in \hat \S^{ij}\mid \hat \phi_{n,k} (\hat \mu_{n,k}(\hat \s_n)) = \tilde \s^{ij}_{n,k}, \text{ for each $n,k$ }  \}$. The set $\hat C^{ij}$ is then a polytope, since $\hat \phi_{n,k} \circ \hat \mu_{n,k}$ is affine, for $k=0,1$. Observe that $\hat C^{ij} = \prod_{n}\hat C^{ij}_n$, where $\hat C^{ij}_n = \{ \hat \s_n \in \hat \S^{ij}_n \mid \hat \phi_{n,k} (\hat \mu_{n,k}(\hat \s_n)) = \tilde \s^{ij}_{n,k}, \text{ for each $k=0,1$} \}$.

By Property (2) of Lemma \ref{lem triangulation T tilde}, each extreme point of $\hat C_n^{ij}$ belongs to a face of $\hat{\S}_n$ of dimension $\text{dim}(\S_n) + \text{dim}(\S_{n+1})$ (see Example \ref{expolytopecomponent}).  Let $\hat \s^{ij}$ be one of the extreme points of the polytope.  For each $n$, let $\hat S_n^{ij, *}$ be the support of $\hat  \s^{ij}$. 

\begin{example}\label{expolytopecomponent}
Figure \ref{illustcomponent}  illustrates the set $\hat C^{ij}_n$ and the effects of property (2) of Lemma \ref{lem triangulation T tilde}. The vertices of the square in (b) of Figure \ref{lemtriangT} are declared to be pure strategies for player $n$ in $\hat G$, as we described above, yielding a tetrahedron as a mixed strategy set, like in Figure \ref{illustcomponent}. Because of property (2) of Lemma \ref{lem triangulation T tilde}, the end-points of the resulting component of equilibria  (one illustrated in red, and the other in blue), are contained in the interior of faces of dimension equal to dim($\S_n$). Without property (2) of Lemma \ref{lem triangulation T tilde}, it could be the case that such end points would fall in faces of lower dimension, which could restrict the index number we would like to assign to equilibria. This point will be addressed in the last step. 

\begin{figure}
\caption{Illustration of $\hat C^{ij}_n$}\label{illustcomponent}

\begin{tikzpicture}[scale=3]

\coordinate (A) at (0, 0.3, 0);
\coordinate (B) at (2, 0, 0);
\coordinate (C) at (1, 1.5, 0);
\coordinate (D) at (1, 0.5, 1.5);

\draw[thick] (A) -- (B) -- (C) -- (A) -- (D) -- (B);

\filldraw[black] (A) circle (0.5pt) node[anchor=north east] {A};
\filldraw[black] (B) circle (0.5pt) node[anchor=north west] {B};
\filldraw[black] (C) circle (0.5pt) node[anchor=south] {C};
\filldraw[black] (D) circle (0.5pt) node[anchor=south east] {D};

\coordinate (BarycenterADC) at (barycentric cs:A=1,D=1,C=2);  
\filldraw[red] (BarycenterADC) circle (0.5pt) node[anchor=north west] {$\hat \sigma^{ij}_n$};

\coordinate (BarycenterDCB) at (barycentric cs:D=2,C=1,B=2);  
\filldraw[blue] (BarycenterDCB) circle (0.5pt) node[anchor=north west] {$\hat \sigma^{ij}_n$};

\draw[thick,red] (BarycenterADC) -- (BarycenterDCB);

\draw[thick] (D) -- (C);

\end{tikzpicture}

\end{figure}
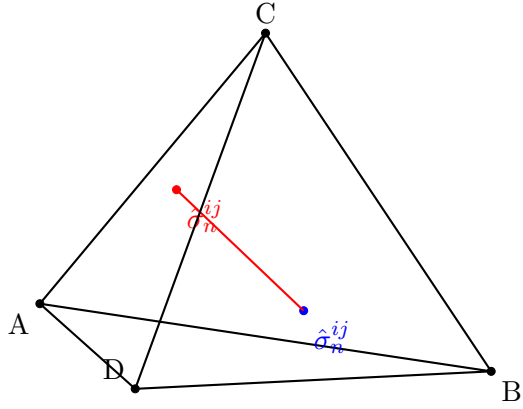
\end{example}

We will now construct an $\e$-perturbation of $\hat G$. In the next and final step, we show that the equilibria of this perturbed game are the points $\hat \s^{ij}$ and that the index of each $\hat \s^{ij}$ is $r^{ij}$.

The perturbation of the payoffs of $\hat G$ for each $n$ has five components: $\hat G_n^0: \hat \S \to \Re_{++}$; $\hat G^1_n: \hat \S \to \Re_+$; $\hat G^*_n: \hat \S \to \Re_{-}$; $\hat g_n^0: \hat S_{n,0} \to \Re$; and $\hat g_n^1: \hat S_{n,1} \to \Re$. For $\hat s \in \hat S$, with $\hat s_m = (\hat s_{m,0}, \hat s_{m,1})$ for each $m$:
\[
\begin{array}{rcl}
\hat G_n^0(\hat s) & = & \hat \phi_{n,0}(\hat s_{n,0})\cdot\tilde g^0_{n}(\hat \phi_{-n,0}(\hat s_{-n, 0}), \hat \phi_{n-1, 1}(\hat s_{n-1,1})) \\  
\hat G_n^1(\hat s) & =  & \hat \phi_{n,1}(\hat s_{n,1})\cdot\tilde g^1_{n}(\hat \phi_{n+1, 0}(\hat s_{n+1,0}))\\  
\hat G_n^*(\hat s) & = & - \mathbbm{1}_{[\cup_{i,j} (\hat S_n^{ij} \backslash \hat S_{n}^{ij, *}) \times \hat S_{-n}^{ij}]}\\
\hat g_n^0(\hat s_n) & =  &-\g_n(\hat \phi_{n,0}(\hat s_{n,0})) \\
\hat g_n^1(\hat s_n) & =  &-\g_{n+1}(\hat \phi_{n,1}(\hat s_{n,1})).
\end{array}  
\]
The function $\hat G_n^0$ depends on $\hat \s_{-n,0}$ and $\hat \s_{n,0}$ (the marginal of strategies $\hat \s_n$ on $\hat \S_{n,0}$); the function $\hat G^1$ depends on $\hat \s_{n+1, 0}$ and $\hat \s_{n,1}$.   The function $\hat G_n^1$ is affine on subsets of $\hat \S_{n+1,0}$ that project under $\hat{\phi}_{n+1,0}$ to simplices of $\tilde \T_{n+1}$: recall from item (1) in Lemma \ref{lem affine} that this is the case in the subcomplex $\mathcal{S}^*_{n+1,0}$, and everywhere else $\tilde g^1_{n, \tilde s_{n,1}}$ is identically $0$.  These two functions provide a multilinear approximation of the continuous perturbations $\tilde g_n^0$ and $\tilde g_n^1$: they are linear in the strategies of each player $m$.  Observe that two strategies $\hat \s_n$ of player $n$ that have the same marginals $\hat \s_{n, 0}$ and $\hat \s_{n,1}$ are equally good replies if we use only these perturbations $\hat G_n^0$ and $\hat G_n^1$: in particular there are a lot of equilibria of the game that are equivalent to the equilibria $\s^{ij}$. In order to eliminate this multiplicity, we introduce three other components in the perturbation.

The function $\hat G^*$ penalizes $n$'s use of a strategy in $\cup_{i,j} \hat S_n^{ij}$ that does not belong to the support of $\hat \s_n^{ij}$ (i.e., $\hat S^{ij,*}_n$) when $n$'s opponents are playing a strategy in $\hat S_{-n}^{ij}$.  Even with all other perturbations, there is a multiplicity of equilibria equivalent to $\s^{ij}$ as those arguments explicitly depend on the marginals on $\S_n^0$ and $\S_n^1$, rather than the normal-form strategies $\hat \S_n$. This penalty helps to isolate  a unique profile among those that are equivalent to $\s^{ij}$. 

The functions $\hat g_n^0$, $\hat g_n^1$ are bonus functions like we constructed for the game $\tilde G$. Their purpose is to eliminate equivalent equilibria that arise when we introduce duplicate strategies as the vertices of triangulation. We illustrate this in Example \ref{final} below. 

\begin{example}\label{final}
Recall that $\g_n$ is a piecewise linear convex function that is linear precisely in the cells of the triangulation $\tilde \T_n$ (cf. property (3) of Lemma \ref{lem triangulation T tilde}). Suppose player $n$ has a $1$-simplex as mixed strategy set, represented in the $x$-axis in Figure \ref{therole}. This simplex is triangulated as indicated, i.e., partitioned in intervals, and the vertices introduced in this triangulation are duplicates of mixed strategies in an equivalent game. Once the vertices become duplicates, player $n$ has mutiple ways of playing the same mixed strategy: for example, in order to play the strategy corresponding to $0.5$ in the figure, player $n$ can either play $0$ with probability $1/2$ and $1$ with probability $1/2$, or $0.25$ with probability $1/2$ and $0.75$ with probability $1/2$, or even $0.5$ with probability $1$. All these strategies induce the same payoffs against any strategy of the opponents in the equivalent game. However, when we add to the payoffs of player $n$ a piecewise linear convex penalty, as illustrated by the graph of the function in red, then the strategy with the lowest penalty associated is $0.5$. In general, for any strategy $\s_n$ of the $1$-simplex, among the many convex combinations that can induce this strategy in the equivalent game, the vertices of the carrier of $\s_n$ in the triangulation are the strategies with lowest associated penalties in the equivalent game. So, if $\s_n$ is a best reply to some opponents' profile in the original game, adding $\hat g^k_n$ to the payoffs in the equivalent game allows us to select the support of the best reply (in the equivalent game) among all strategies equivalent to $\s_n$.\footnote{The same technique has been used to select supports of best replies in \cite{GW2005} and \cite{GLP2023}.}

\begin{figure}
\caption{The role of $\hat g^{k}_n, k=0,1$ }\label{therole}
\begin{tikzpicture}[scale=4]

\draw[thick, ->] (0,0) -- (1.1,0) node[anchor=north] {$x$};

\draw[thick, ->] (0,0) -- (0,1.2) node[anchor=east] {$\g^{k}_n$};

\foreach \x in {0, 0.25, 0.5, 0.75, 1} {
    \draw (\x, 0) -- (\x, -0.02);  
    \node at (\x, -0.05) {\x};  
}

\draw[thick, red] (0, -0.5+1.2) -- (0.25, -0.85+1.2) -- (0.5, -1+1.2) -- (0.75, -0.85+1.2) -- (1, -0.5+1.2);

\filldraw[red] (0, -0.5+1.2) circle (0.5pt);
\filldraw[red] (0.25, -0.85+1.2) circle (0.5pt);
\filldraw[red] (0.5, -1+1.2) circle (0.5pt);
\filldraw[red] (0.75, -0.85+1.2) circle (0.5pt);
\filldraw[red] (1, -0.5+1.2) circle (0.5pt);

\end{tikzpicture}

\end{figure}
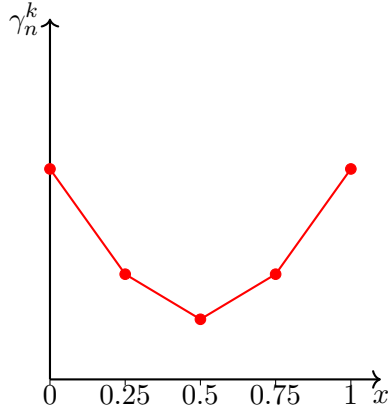
\end{example}

For positive constants $\a$ and $\a^*$, define the  game $\hat G^{\a, \a^*} \equiv \hat G \oplus \hat G^0 \oplus \hat G^1 \oplus \a^*\hat G^* \oplus \a \hat g^0 \oplus \a\hat g^1$.  Fix $\a>0$ sufficiently small such that $\hat G^{\a, 0}$ is an $\e$-perturbation.  If $\a^*$ is chosen small as well then $\hat G^{\a, \a^*}$ is also an $\e$-perturbation of $\hat G$. 

\medskip

\noindent {\bf Step 9.}  Let us first analyze the game $\hat G^{\a, 0}$, i.e, the perturbation with $\a^* = 0$.   Suppose $\hat \phi(\mu(\hat \s)) = \tilde \s^{ij}$ for some $ij$. We claim that $\hat \s$ is an equilibrium iff $\hat \s \in \hat C^{ij}$.  If $\hat \s_n  \notin \hat C_n^{ij}$ for some $n$, then a direct computation shows that $\hat \s_n^{ij}$ gives: (1) the same payoff as $\hat \s_n$ in the components $\hat G  \oplus \hat G^0 \oplus \hat G^1$;  and (2) a strictly higher payoff in $\hat g^0 \oplus \hat g^1$.  Thus, if $\hat \s$ is an equilibrium then $\hat \s \in \hat C^{ij}$.  Conversely, a similar computation shows that if $\hat \s \in \hat C^{ij}$, then any pure strategy that is not in $\hat S_n^{ij}$ is a strictly inferior reply against $\hat \s$, and $\hat{\s}$ is an equilibrium of $\hat G^{\a, 0}$.   

We will now show that the $\hat C^{ij}$'s are the only equilibria of the game $\hat G^{\a, 0}$ (note that in the previous paragraph we showed that they are the only equilibria of $\hat G^{\a, 0}$ among those profiles $\hat{\s}$ for which $\hat \phi(\mu(\hat \s)) = \tilde \s^{ij}$). Let $\hat \s$ be an equilibrium of $\hat G^{\a,0}$. For each $n$, let $\tilde \s_{n,k} = \hat \phi_{n,k}(\hat \mu_{n,k}(\hat \s_n))$ for $k = 0, 1$. It is sufficient to show that  $\tilde \s =\tilde \s^{ij}$ for some $i,j$.

Any strategy $\hat \t_n$ of player $n$ in $\hat G^{\a, 0}$  such that  $\tilde \s_{n,k} = \hat \phi_{n,k}(\hat \mu_{n,k}(\hat \t_n))$ for each $k$ yields the same payoff in $\hat G$, $\hat G^0$, and $\hat G^1$.  The choice of which among these strategies is optimal depends on the payoffs they obtain under $\hat g^k_{n}$ for $k = 0, 1$. The optimality of $\hat \s_n$ implies that $\hat \s_{n,0}$ is a mixture over the vertices of the simplex of $\tilde \T_{n}$ that contains $\tilde \s_{n-1,1}$ in its interior and $\hat \s_{n,1}$ is a mixture over the vertices of the simplex of $\tilde \T_{n+1}$ that contains $\tilde \s_{n+1,0}$ in its interior (cf. with our discussion in Example \ref{final}). Thus $\tilde \s$ is an equilibrium of the game $\tilde G \oplus g^0 \oplus g^1$ where:
\[
\begin{array}{rcl}
g_{n, \tilde s_{n,0}}^0 & = & \sum_{\tilde s_{-n}}\prod_{m \neq n}\tilde \s_{m}(\tilde s_m)\tilde g^0_{n, \tilde s_{n,0}}(\tilde s_{-n, 0}, \tilde s_{n-1, 1}) \\
g_{n, \tilde s_{n,1}}^1 &  = & \sum_{\tilde s_{n+1}}\tilde \s_{n+1}(\tilde s_{n+1})\tilde g^1_{n, \tilde s_{n,1}}(\tilde s_{n+1, 0}).
\end{array}
\]
Since for each $k=0,1$, $\Vert g^k - \tilde g^k(\tilde \s) \Vert < \xi$, (A) of step 7 implies that $\tilde \s$ is in $\tilde T^{ij} \setminus \partial \tilde T^{ij}$ for some $i,j$.  If  $\tilde \s \neq \tilde \s^{ij}$, there is some $n$ for which either (a) $\tilde \s_{n,0} \neq \tilde \s_{n,0}^{ij}$ or (b) $\tilde \s_{n,1} \neq \tilde \s_{n,1}^{ij}$. Case (a) implies that  $\tilde \s_{n-1,1}$ belongs to the boundary of $T_{n-1,1}^{ij}$ by definition of $\tilde {g}^1_{n-1}$ and the fact it is affine over $\tilde T^{ij}_{n,0}$, which is a contradiction with our assumption that $\tilde \s$ is in $\tilde T^{ij} \setminus \partial \tilde T^{ij}$. Hence, (a) does not occur, and it must be (b). The payoffs from $\tilde \s$ to player $n+1$ in $\tilde G \oplus g^0 \oplus g^1$ are then $\sum_{\tilde s_{-n}}\prod_{m \neq n}\tilde \s_{m}(\tilde s_m)r_n(\phi_{n,0}(\tilde s_{-n,0})) + \tilde \s_{n+1,0} \cdot \e_0 \tilde f_{n}(\tilde \s_{n,1}) + \tilde \s_{n+1,1} \cdot g^1_n$. By construction, $r_n$ is constant on $\tilde T^{ij}_{-n}$ (cf. step 5), which implies it is $\tilde f_{n+1,0}(\tilde \s_{n,1})$ that decides which $\tilde \s_{n+1,0}$ is a best reply. But since $\tilde f_{n+1,0}(\tilde \s_{n,1}) \neq \tilde \s^{ij}_{n+1}$ (cf. (3) of Lemma \ref{lem f ij}), it follows that $\tilde \s_{n+1,0}$ belongs to the boundary of $\tilde \beta_{n+1,0}(T^{ij}_{n+1})$, which is again a contradiction. Therefore, $\tilde \s = \tilde \s^{ij}$. From the claim proved at the beginning of this step, it follows that the $\hat C^{ij}$'s are the only equilibria of the game $\hat G^{\a, 0}$, as we wanted to prove. 


Now we turn to the equilibria of $\hat G^{\a, \a^*}$.  For each $ij$, $\hat \s^{ij}$ is an isolated equilibrium of $\hat G^{\a, \a^*}$ because of the penalty function $\hat G^*$.  We will show that the index of $\hat \s^{ij}$ is $r^{ij}$.  Since the strategies not in $\hat \S^{ij*}$ are inferior replies to $\hat \s^{ij}$, the index can be computed in the game obtained by deleting them.  Recall also that $\hat g^0_n$ and $\hat g^1_n$ are constant over $\hat \S_n^{ij}$.  Therefore, the payoff of player $n$ in $\hat G \oplus \hat G^0 \oplus \hat G^1$ restricted to $\hat \S_n^{ij}$ from any strategy $(\hat s_{n,0}, \hat s_{n,1})$ is an affine function of $\hat \s_{-n}$ equal to $D + \hat \phi_{n,0}(\hat s_{n,0}) \cdot \tilde f_{n,0}(\hat \phi_{n-1,1}(\hat \mu_{n-1,1}(\hat \s_{n-1}))) + \hat \phi_{n,1}(\hat s_{n,1}) \cdot \tilde f_{n,1}(\hat \phi_{n+1,0}(\hat \mu_{n+1,0}(\hat \s_{n+1})))$, where the constant $D$ comes from the fact that $r_n$ is constant in $\hat \S^{ij}$. Therefore, the best reply correspondence of this game assigns to $\hat \s^{ij}$ the same index as $\tilde \varphi$ assigns to $\tilde \s^{ij}$.

To finish the proof, it remains to be shown that for small $\a^* > 0$, the only equilibria of $\hat G^{\a, \a^*}$ are the $\hat \s^{ij}$'s. Let $\hat \s$ be an equilibrium of $\hat G^{\a, \a^*}$ where $\a^* > 0$ and let $\tilde \s = \hat \phi \circ \hat \mu(\hat \s)$.   If $\a^*$ is sufficiently small,  then $\tilde \s$ is close to $\tilde \s^{ij}$ for some $ij$ and there exists $\hat \s^*$ with support $S^{ij*}$ such that $\hat \phi \circ \hat \mu(\hat \s^*) = \tilde \s$. Also, any strategy that is not in $S_n^{ij}$ for player $n$ is a inferior reply to $\hat \s$, as this property holds when $\a^* = 0$ and the equilibrium is in $\hat C^{ij}$. This implies that $\hat{\s}_n$ has support in $\hat{S}^{ij}_n$, for each $n$. Because of the penalty function $\hat G_n^*$, the support of $\hat{\s}_n$ must be in particular in $\hat{S}^{ij*}_n$. Therefore, $\hat \s^* = \hat \s$. The restriction of the game to $S^{ij*}$ is a polymatrix game and as such the set of its equilibria with support $S^{ij*}$ is a convex set.  As $\hat \s^{ij}$ is an isolated equilibrium with support $S^{ij*}$, $\hat \s^{ij} = \hat \s$ and the proof is complete.

\medskip

\setcounter{secnumdepth}{-1}

\section{Appendix: Proof of Lemma \ref{lem triangulation T tilde}}

 Let $m \equiv$ dim$(\D(\tilde S_{n,0}))$, so the number of vertices of $\D(\tilde S_{n,0})$ is $m+1$. Recall the definition of $\T_n$ in Lemma \ref{lem triangulation T}. Note that $\T_n$ can be assumed to have any finite number of vertices. Letting $K_n \equiv \#\{\tilde \s^{ij}_n\}_{i,j}$, we will therefore assume that $K_n[\text{dim}(\S_n) + 1]< m+1$.

A polyhedral subdivision $\P_n$ of $\D(\tilde S_{n,0})$ is said to be \textit{regular} (cf. \cite{LRS2010}) if there exists a \textit{height function} $h: P_n \to \Re$ from the vertex set $P_n$ of $\P_n$, where the vertex set $\{p_t\}_{t}$ of a maximal-dimensional cell of the subdivision is mapped to points in $\Re$ s.t. $\{(p_t, h(p_t))\}_t$ spans a (non-vertical) hyperplane in $\D(\tilde S_{n,0}) \times \Re$ and all other points $p \in P_n$ are such that $(p, h(p))$ is strictly above the hyperplane. In this case, we say that the set  $\{p_t\}_{t}$ is \textit{lifted} by $h$ \textit{to a hyperplane}, with all other vertices \textit{lifted above} the hyperplane. 

Let $\tilde \T_n$ be a regular triangulation with vertex set $T_n$. For each $p \in T_n$, let $v_p \in \D(\tilde S_{n,0})$ be any point with the same carrier as $p$ in the face complex of $\D(\tilde S_{n,0})$.  For $\e \geq 0$, let $p(\e) \equiv (1-\e)p + \e v_p$. For $\e>0$ sufficiently small, define the triangulation $\T^{\e}_n$ as follows: $B(\e) \equiv [p_0(\e),...,p_{m}(\e)]$ is a maximal dimensional cell of $\tilde \T^{\e}_n$, given by the convex hull of $\{p_0(\e),...,p_m(\e)\}$ iff $B \equiv [p_0,...,p_{m}]$ is a maximal dimensional cell in $\tilde \T_n$ given by the convex hull of $\{p_0,...,p_m\}$.



\begin{lemma}\label{regularisgeneric}
Suppose $\tilde \T_n$ is a regular triangulation. There exists $\e>0$ sufficiently small such that any triangulation $\tilde \T^{\e}_n$ is regular.
\end{lemma}

\begin{proof} 
Take $\{p_0,....,p_{m}\}$ to be the vertex set of a maximal dimensional cell $B$ of $\tilde \T_n$ and let $p \in T_n \setminus B$. Since $\tilde \T_n$ is regular, let its associated height function be denoted by $h$. 
Then $(p, h(p))$ lies above the hyperplane defined by $(p_t, h(p_t))^{m}_{t=0}$. Therefore, substituting $p_t$ for $p_t(\e)$ and letting $\e>0$ be sufficiently small, we have that $(p_t(\e), h(p_t(\e)))^{m}_{t=0}$ defines a non-vertical hyperplane in $\Re^{m+2}$ with $(p(\e), h(p(\e))$ above this hyperplane, for any other vertex $p(\e)$. From this, we can define the graph of a convex piecewise linear function $h^{\e}: \D(\tilde S_{n,0}) \to \Re$, by letting $h^{\e}(p) \equiv h(p)$ for every vertex $p$ of $\T^{\e}_n$, and, for another arbitrary point $p$, we consider the carrier of $p$ in $\tilde \T^{\e}_n$ and define $h^{\e}$ by linear interpolation.  \end{proof}

We now define a \textit{generalized barycentric subdivision} of the simplex $\D(\tilde S_{n,0})$, which is a minor generalization of the classic barycentric subdivision. Let $\tilde{\T}^0_n$ be the face complex of $\D(\tilde S_{n,0})$. Assume $\tilde{\T}^{k-1}_n$ is defined and let us define $\tilde{\T}^{k}_n$: the \textit{$0$-dimensional generalized barycenters} are the vertices of $\tilde{\T}^{k-1}_n$. For each $1$-dimensional cell $\tilde \t$ of $\tilde{\T}^{k-1}_n$, chose exactly one point $b^1$ in its interior. The point $b^1$ is referred to as a \textit{$1$-dimensional generalized barycenter} of $\tilde \t$. One then proceeds by choosing points in the interior of cells of increasing dimension, in the exact same fashion as with the classical barycentric subdivision.

The full-dimensional cells of $\tilde{\T}^{k}_n$ are also defined in the exact same fashion as in the $k$-th-iterate of the classical barycentric subdivision: an $m$-dimensional cell of $\tilde{\T}^{k}_n$ is denoted by its vertices $(b^0_k, b^1_k,...,b^m_k)$, where $b^0_k$ is a vertex of $\tilde{\T}^{k-1}_n$, $b^1_k$ a generalized barycenter of a $1$-dimensional cell of $\tilde{\T}^{k-1}_n$,  $b^2_k$ a generalized barycenter of a $2$-dimensional cell of $\tilde{\T}^{k-1}_n$, etc. When generalized barycenters are chosen sufficiently near the actual barycenters, the diameter of a generalized barycentric subdivision shrinks, as $k$ increases, just like with the classical barycentric subdivision.


Letting $\zeta>0$ be as in the statement of Lemma \ref{lem triangulation T tilde}, there exists $\d>0$ sufficiently small and $k_0 \in \mathbb{N}$ such that for any $k \geq k_0$, any choice of generalized barycenters within $\d$ of the actual barycenters yields a generalized barycentric subdivision $\tilde{\T}^{k}_n$ with diameter less than $\zeta$. Moreover, we can assume without loss of generality that for each $i,j$, the carrier of $\tilde \s^{ij}_n$ in $\tilde{\T}^{k}_{n}$ has all its vertices in the interior of $\tilde{T}^{ij}_n$. We fix from now on such a $\d>0$ and $k_0$.

We now define a polyhedral refinement of $\tilde{\T}^{k}_n$, which we call the \textit{Eaves-Lemke (EL)-refinement $\P_n$ of $\tilde{\T}^{k}_n$}. This is precisely the same polyhedral refinement as the one used in \cite{GW2005}, we recall its construction for completeness. For each $(m-1)$-dimensional cell $\t$ of $\tilde{\T}^{k}_n$, let $H_{\t} \equiv \{ x \in \R^{m+1} \mid a_{\t} \cdot x = b_{\t} \}$ be the hyperplane that includes $\t$ and is orthogonal to $\D(\tilde S_{n,0})$. Each full-dimensional cell of $\P_n$ is the intersection of $\D(\tilde S_{n,0})$ with a polyhedron $\cap_{\t} H^{i}_{\t} $, where $i \in \{+,-\}$ and $H^{i}_{\t}$ is one of the two half-spaces defined by the hyperplane $H_{\t}$.

\begin{lemma}\label{EL}
There exists a polyhedral subdivision $\P_n$ of $\D(\tilde S_{n,0})$ satisfying the following properties: 

\begin{enumerate}
\item for each $i,j$, $\tilde \s^{ij}_n$ has a carrier $\rho^{ij}$ in $\P_n$ of dimension dim($\S_n$);
\item there is a convex piecewise linear function $\g_n: \D(\tilde S_{n,0}) \to [0,1]$ which is linear precisely in each cell of $\P_n$;
\item the diameter of $\P_n$ is less than $\zeta$.
\end{enumerate} 
\end{lemma} 

\begin{proof} 
We first focus on constructing a polyhedral subdivision that satisfies (1). In $\tilde \T^0_n$, choose generalized barycenters which are distinct from $\tilde{\s}^{ij}_n$, for each $i,j$, defining $\tilde{\T}^1_n$. Consider now the derived (EL)-polyhedral subdivision $\P^1_ n$ of $\tilde \T^{1}_n$. If there exists $i,j$ such that the carrier of $\tilde{\s}^{ij}_n$ in $\P^1_n$ has dimension less than dim($\S_n$), then there exists a hyperplance $H_{\t}$ which contains $\tilde{\s}^{ij}_n$, where $\t$ is a face containing an $m$-dimensional generalized barycenter $b^m_1$. Let $\{p_0, p_1, ..., p_{m-2}, b^{m}_{1}\}$ be affinely independent vertices of the triangulation $\tilde{\T}^{1}_n$ defining $H_{\t} \cap \D(\tilde{S}_{n,0})$. If $H_{\t}$ contains in addition $\tilde{\s}^{ij}_n$, then $\{p_0, p_1, ..., p_{m-2}, b^{m}_{1}, \tilde{\s}^{ij}_n\}$ is an affinely dependent set. Choosing therefore $b^{m}_{1}$ outside an affine set with dimension strictly lower than $m$ in $\D(\tilde{S}_{n,0})$ implies $H_{\t}$ does not contain $\tilde{\s}^{ij}_n$. This procedure can be iteratively applied for each point $\tilde{\s}^{ij}_n$, ensuring all of them have a carrier of dimension dim($\S_n$) in $\P^1_n$. Suppose now we have obtained $\P^{k_0-1}_n$ from $\tilde{\T}^{k_0-1}_n$ satisfying the property that all $\tilde \s^{ij}_n$ have carriers with dimension equal to dim($\S_n$). Choose generalized barycenters in $\tilde{\T}^{k_0-1}_n$ such that each generalized barycenter is distinct from each $\tilde \s^{ij}_n$ and consider the (EL)-polyhedral subdivision $\P^{k_0}_n$. If there exists $i,j$, such that the carrier of $\tilde \s^{ij}_n$ in $\P^{k_0}_n$ has dimension less than dim($\S_n$), then it follows that there exists $\t$ containing an $m$-dimensional generalized barycenter $b^m_{k_0}$ such that $H_{\t}$ contains $\tilde{\s}^{ij}_n$. Similarly as before, let $\{p_0, p_1, ..., p_{m-2}, b^{m}_{k_0}\}$ be affinely independent vertices of the triangulation $\tilde{\T}^{k_0}_n$ defining $H_{\t} \cap \D(\tilde{S}_{n,0})$. If $H_{\t}$ contains in addition $\tilde{\s}^{ij}_n$, then $\{p_0, p_1, ..., p_{m-2}, b^{m}_{k_0}, \tilde{\s}^{ij}_n\}$ is an affinely dependent set. Choosing therefore $b^{m}_{k_0}$ outside an affine set with dimension strictly lower than $m$ in $\D(\tilde{S}_{n,0})$ implies $H_{\t}$ does not contain $\tilde{\s}^{ij}_n$. This procedure can be iteratively applied for each point $\tilde \s^{ij}_n$, ensuring all of them have a carrier of dimension dim($\S_n$) in $\P^{k_0}_n$. This shows $\P^{k_0}_n$ satisfies (1). Since $\P^{k_0}_n$ refines $\tilde{\T}^{k_0}_n$, its diameter is also less than $\zeta$, implying (3).  To show (2), consider the following function: $\g_n: \D(\tilde S_{n,0}) \to [0,1]$, $\g_n(\s_n) \equiv \a \sum_{\t}|a_{\t} \cdot \s_n - b_{\t}|$, where the scaling factor $\a>0$ is chosen so as to let $\g_n(\D(\tilde S_{n,0})) \subseteq [0,1]$. The function $\g_n$ then satisfies (2). \end{proof} 

We now derive a triangulation from $\P_n$ in Lemma \ref{EL}, by triangulating each full-dimensional cell of $\P_n$. This triangulation will add no vertices beyond those of $\P_n$, i.e., it will simply subdivide each $m$-dimensional cell of $\P_n$ into $m$-dimensional simplices, without adding new vertices to the triangulation. This triangulation, up to an arbitrarily small perturbation of its vertices, achieves the objectives of Lemma \ref{lem triangulation T tilde}.

\begin{lemma}\label{final refinement}
There exists a triangulation $\tilde{\T}_n$, which refines $\P_n$ without adding new vertices, satisfying the following conditions: 
\begin{enumerate} 
\item The diameter of $\tilde{\T}_n$ is less than $\zeta$;
\item For  $i,j$, $\tilde T^{ij}_n$ is the space of a subcomplex of $\tilde \T_n$;
\item For  $i,j$, the carrier $\hat{T}^{ij}_n$ of $\tilde \s^{ij}_n$ in $\tilde \T_n$ has dimension dim($\S_n$);
\item There is a convex piecewise linear function $\g_n: \D(\tilde S_{n,0}) \to \Re_{+}$ such that: 
	\begin{enumerate}
	\item[(4.1)] $\g_n$ is linear precisely on the simplices of $\tilde{\T}_n$;
	\item[(4.2)] $\g_n$ is constant in each $\hat T^{ij}_n$;
	\end{enumerate}

\item For each $i,j$, $(\tilde \s^{ij}_n, \tilde \s^{ij}_{n+1})$ is not in the convex hull of dim($\S_n$) $+$ dim($\S_{n+1}$) (or less) vertices of $\hat{T}^{ij}_n \times \hat{T}^{ij}_{n+1}$.
\end{enumerate} 
\end{lemma}

\begin{proof} The polyhedral complex $\P_n$ refines the generalized barycentric subdivision $\tilde{\T}^{k_0}_n$. Therefore it also has diameter less than $\z$ and has $\tilde T^{ij}_n$ as the space of a subcomplex. There exists a piecewise-linear, convex function $\hat{\g}_n: \D(\tilde{S}_{n,0}) \to \Re_{+}$, which is linear precisely in the cells of $\P_n$ (cf. Lemma \ref{EL}). The map $\hat{\g}_n(\cdot)$ is in particular a height function in the set of vertices of $\tilde{\T}^{k}_n$. By Lemma 2.3.15 in \cite{LRS2010},  any sufficiently small (generic) perturbation $\g'_n$ of this height function gives a refinement $\tilde{\T}_n$ of $\P_n$ (without adding vertices to $\P_n$), where $\tilde{\T}_n$ is a triangulation, yielding a convex piecewise linear function $\g'_n: \D(\tilde S_{n,0}) \to \Re_{+}$ which is linear precisely in the cells of $\tilde{\T}_n$. 

We show that it is without loss of generality to assume that the carrier of $\tilde{\s}^{ij}_n$ in $\tilde{\T}_n$ is a full-dimensional cell of $\tilde{\T}_n$. If there exists $i,j$ such that $\tilde{\s}^{ij}_n$ is contained in an $(\text{dim}(\S_n)-1)$-dimensional cell $\t$ of $\tilde{\T}_n$, using Lemma \ref{regularisgeneric}, we can consider a sufficiently small $\e>0$ and a $\tilde \T^{\e}_n$ in which the carrier of $\tilde \s^{ij}_n$ in $\tilde \T^{\e}_n$ has dimension dim($\S_n$): this is done by choosing suitable perturbations $p(\e)$ of vertices $p$ of the carrier of $\tilde \s^{ij}_n$ in $\tilde \T_n$, yielding  the desired $\tilde \T^{\e}_n$. Moreover, for a sufficiently small $\e$, $\tilde \T^{\e}_n$ satisfies (1), (2), (3) and (4.1). Therefore, if necessary by considering a small perturbation $\tilde \T^{\e}_n$, we can assume that $\tilde \T_n$ satisfies (1), (2), (3) and (4.1).


Note now that the convex hull of any collection of dim($\S_n$) $+$ dim($\S_{n+1}$) (or less) vertices of $\hat{T}^{ij}_n \times \hat{T}^{ij}_{n+1}$ is a union of convex sets whose dimensions are strictly less dim($\S_n$) + dim($\S_{n+1}$). This implies that if the vertices of $\hat{T}^{ij}_n \times \hat{T}^{ij}_{n+1}$ are positioned in $\tilde T^{ij}_n$ outside a set of dimension strictly less than dim($\S_n$) + dim($\S_{n+1}$), then $\tilde{\s}^{ij}_n$ satisfies property (5). By using Lemma \ref{regularisgeneric} just like in the previous paragraph, we can then assume $\tilde \T_n$ satisfy (5).

 To finish our construction it remains to prove that we can construct a piecewise linear convex function $\g_n: \D(\tilde S_{n,0}) \to \Re_{+}$, linear precisely on each cell of $\tilde{\T}_n$ and constant over each $\hat{T}^{ij}_{n}$. Since $K_n [\text{dim}(\S_n) + 1]< m+1$, the union of the set of vertices of $\hat T^{ij}_n$ over $i,j$ spans an affine set with dimension lower than $m$. This implies we can define an affine function $\phi_n : \D(\tilde S_{n,0}) \to \Re$ satisfying the following conditions: (a) $\phi_n (\hat {p}^{ij}_n) = \hat{\g}_n(\hat p^{ij}_n)$ , for all vertices $\hat p^{ij}_n$ of $\hat T^{ij}_n$; (b) $\phi_n (p_n) < \hat{\g}_n(p_n)$ for all remaining vertices $p_n$ of $\tilde \T_n$. Letting now $\g_n \equiv \g'_n - \phi_n $ yields the desired function satisfying all the required properties in the statement of the Lemma.  \end{proof}


\end{document}